\documentclass[journal,12pt,onecolumn,draftcls]{IEEEtran}
\IEEEoverridecommandlockouts

\usepackage[noadjust]{cite}
\usepackage{amsmath,amssymb,amsfonts}
\usepackage{algorithmic}
\usepackage{graphicx}
\usepackage{textcomp}
\usepackage{enumitem}
\usepackage{xcolor}

\usepackage{amsthm}
\usepackage{changepage}
\usepackage{mathtools}
\usepackage{nicefrac}
\usepackage{balance}
\usepackage{stfloats}
\usepackage{float}
\usepackage[ruled,lined]{algorithm2e}
\SetAlgoNlRelativeSize{-2}
\definecolor{NYUlight}{HTML}{8900e1} 	

\newtheorem{thm}{Theorem}
\newtheorem{lem}{Lemma}
\newtheorem{prop}{Proposition}

\theoremstyle{definition}
\newtheorem{defn}{Definition}

\newtheorem{rem}{Remark}
\def\BibTeX{{\rm B\kern-.05em{\sc i\kern-.025em b}\kern-.08em
    T\kern-.1667em\lower.7ex\hbox{E}\kern-.125emX}}

\usepackage{bbm}

\newcommand{\blue}[1]{{\color{blue}#1}}
\newcommand{\indep}{\raisebox{0.05em}{\rotatebox[origin=c]{90}{$\models$}}}
\usepackage{etoolbox}    
\usepackage{dirtytalk}

\DeclareMathOperator*{\argmax}{\arg\!\max}
\DeclareMathOperator*{\argmin}{\arg\!\min}
\SetKwComment{Comment}{/* }{ */}
\SetKwInOut{Input}{Input}
\SetKwInOut{Output}{Output}

\SetCommentSty{mycommfont}

\newfloat{topalgorithm}{t}{lop}
\floatname{topalgorithm}{Algorithm}

\begin{document}

\title{Distribution-Agnostic Database De-Anonymization Under Obfuscation And Synchronization Errors\\
\thanks{This work was presented in part at the 2023 IEEE International Workshop on Information Forensics and Security (WIFS), Nuremberg, Germany. This work is supported in part by National Science Foundation grants 2148293, 2003182, and 1815821, and NYU WIRELESS Industrial Affiliates.}}

\author{Serhat Bakirtas, Elza Erkip\\
 NYU Tandon School of Engineering\\
Emails: \{serhat.bakirtas, elza\}@nyu.edu }

\maketitle

\begin{abstract}
    Database de-anonymization typically involves matching an anonymized database with correlated publicly available data. Existing research focuses either on practical aspects without requiring knowledge of the data distribution yet provides limited guarantees, or on theoretical aspects assuming known distributions. This paper aims to bridge these two approaches, offering theoretical guarantees for database de-anonymization under synchronization errors and obfuscation without prior knowledge of data distribution. Using a modified replica detection algorithm and a new seeded deletion detection algorithm, we establish sufficient conditions on the database growth rate for successful matching, demonstrating a double-logarithmic seed size relative to row size is sufficient for detecting deletions in the database. Importantly, our findings indicate that these sufficient de-anonymization conditions are tight and are the same as in the distribution-aware setting, avoiding asymptotic performance loss due to unknown distributions. Finally, we evaluate the performance of our proposed algorithms through simulations, confirming their effectiveness in more practical, non-asymptotic, scenarios.
\end{abstract}

\section{Introduction}
\label{sec:introduction}
The accelerating growth of smart devices and applications has accelerated the collection of user-level micro-data both by private companies and public institutions. This data is often shared and/or sold after removing explicit user identifiers, a.k.a. \emph{anonymization}, and coarsening of the data through noise, a.k.a. \emph{obfuscation}. Despite these efforts, there is a growing concern about the privacy implications~\cite{ohm2009broken}. These concerns were further justified by the success of a series of practical de-anonymization attacks on real data~\cite{naini2015you,datta2012provable,narayanan2008robust,sweeney1997weaving,takbiri2018matching}. In the light of these successful attacks, recently there has been an increasing effort on understanding the information-theoretic and statistical foundations of \emph{database de-anonymization}, a.k.a. \emph{database alignment/matching/recovery}~\cite{cullina,shirani8849392,dai2019database,kunisky2022strong,tamir2023correlation,bakirtas2021database,bakirtas2022matching,bakirtas2022seeded,bakirtas2023database,bakirtas2022database}.

Our recent work focuses on the database de-anonymization problem under synchronization errors. In~\cite{bakirtas2022matching}, we investigated the matching of Markov databases under synchronization errors only, with no subsequent obfuscation/noise, using a histogram-based detection method. In~\cite{bakirtas2022seeded,bakirtas2023database}, we extended this effort to databases with noisy synchronization errors. We introduced a noisy replica detection algorithm, and a seeded deletion detection algorithm, and used joint-typicality-based matching to derive achievability results, which we subsequently showed to be tight, for a seed size double logarithmic with the row size of the database. In~\cite{bakirtas2022matching,bakirtas2022seeded,bakirtas2023database} we assumed that the underlying distributions are known and tailored repetition detection and matching algorithms for these known distributions. 

\begin{figure}[t]
\centerline{\includegraphics[width=0.75\textwidth,trim={0cm 9cm 17cm 0},clip]{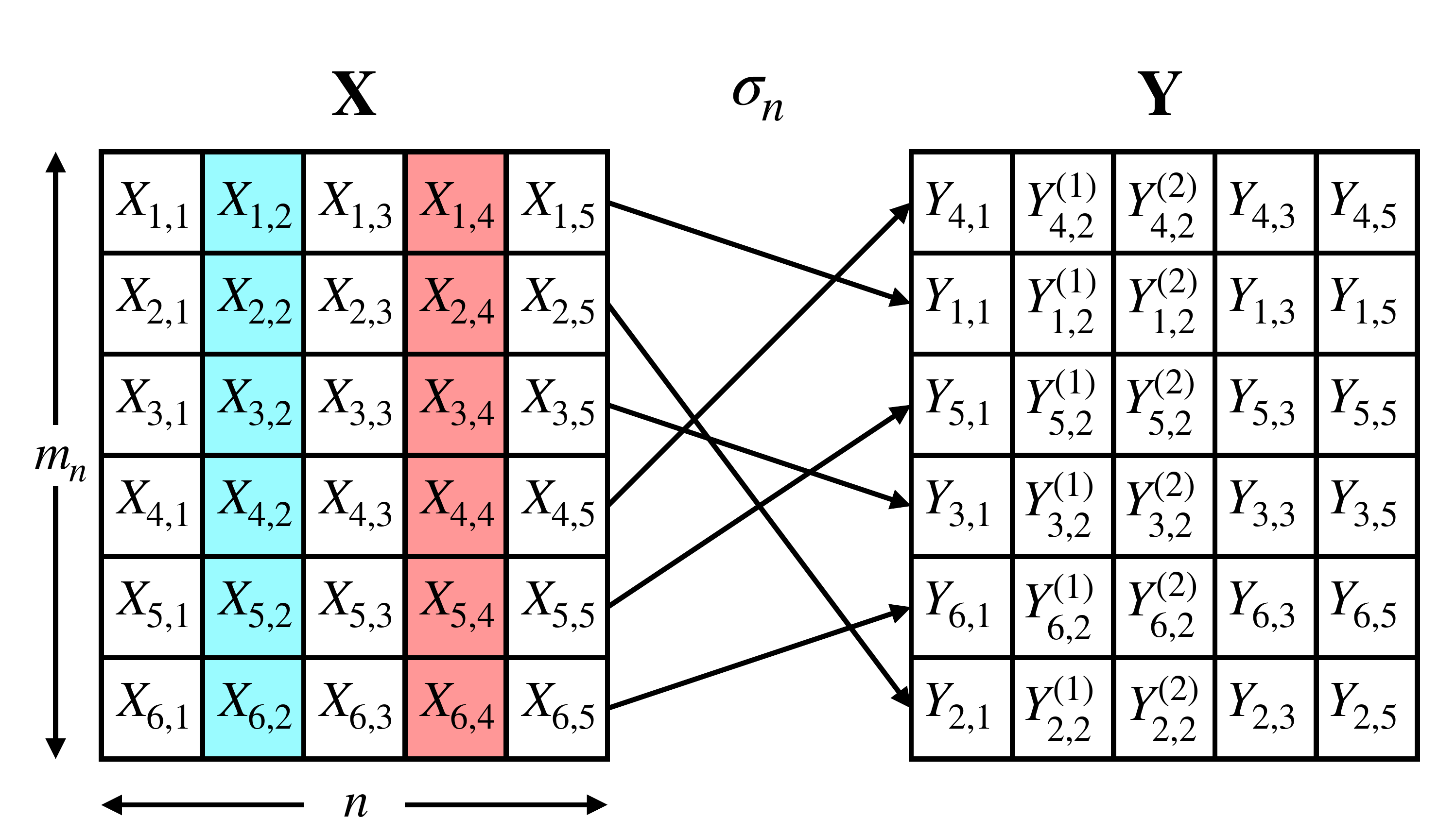}}
\caption{An illustrative example of database matching under column repetitions. The column colored in red is deleted, whereas the column colored in blue is replicated. $Y_{i,2}^{(1)}$ and $Y_{i,2}^{(2)}$ denote noisy copies/replicas of $X_{i,2}$. The goal of database de-anonymization studied in this paper is to estimate the correct row permutation ${\sigma_n=\left(\begin{smallmatrix} 
1 & 2 & 3 & 4 & 5 & 6\\
2 & 6 & 4 & 1 & 3 & 5
\end{smallmatrix}\right)}$, by matching the rows of $\mathbf{X}$ and $\mathbf{Y}$ without \emph{any prior information} on the underlying database ($p_{X}$), obfuscation ($p_{Y|X}$), and repetition ($p_S$) distributions.}
\label{fig:intro}
\end{figure}

Motivated by most practical settings where the underlying distributions are not readily available, in this paper, we investigate the de-anonymization problem \emph{without any prior knowledge of the underlying distributions}. Borrowing a noisy random column repetition model from~\cite{bakirtas2022seeded}, illustrated in Figure~\ref{fig:intro}, we show that even though the distributions are not known a priori, matching can be performed with no asymptotic (in database size) loss of performance. To that end, we first modify the noisy replica detection algorithm proposed in~\cite{bakirtas2022seeded} so that it still works in the distribution-agnostic setting. Then, we propose a novel outlier-detection-based deletion detection algorithm and show that when seeds, whose size grows double logarithmic with the number of users (rows), are available, the underlying deletion pattern could be inferred. Next, through a joint-typicality-based de-anonymization algorithm that relies on the estimated distributions and the repetition pattern, we derive a tight sufficient condition for database de-anonymization to succeed. Finally, we evaluate the performance of our proposed distribution-agnostic algorithms through simulations on finite databases and demonstrate their success in the non-asymptotic regime as well.

The structure of the rest of this paper is as follows: Section~\ref{sec:problemformulation} introduces the formal statement of the problem. Section~\ref{sec:mainresult} contains our proposed algorithms, states our main result, and contains its proof. In Section~\ref{sec:experiments} we evaluate the performances of our proposed algorithms in the non-asymptotic regime via simulations. Section~\ref{sec:conclusion} consists of the concluding remarks. Proofs are provided in the Appendix.

\noindent{\em Notation:} We denote a matrix $\mathbf{X}$ with bold capital letters, and its $(i,j)$\textsuperscript{th} element with $X_{i,j}$. A set is denoted by a calligraphic letter, \emph{e.g.}, $\mathcal{X}$. $[n]$ denotes the set of integers $\{1,\dots,n\}$. Asymptotic order relations are used as defined in~\cite[Chapter 3]{cormen2022introduction}. Unless stated otherwise, all logarithms are base 2. $H(.)$ and $I(.;.)$ denote the Shannon entropy and the mutual information~\cite[Chapter 2]{cover2006elements}, respectively. $D(p\|q)$ denotes the relative entropy~\cite[Chapter 2.3]{cover2006elements} (in bits) between two Bernoulli distributions with respective parameters $p$ and $q$.  $\overset{p}{\to}$ denotes convergence in probability.
\newpage
\section{Problem Formulation}
\label{sec:problemformulation}

We use the following definitions, most of which are adopted from~\cite{bakirtas2023database} to formalize our problem.

\begin{defn}{\textbf{(Anonymized Database)}}\label{defn:unlabeleddb}
An ${(m_n,n,p_{X})}$ \emph{anonymized database} $\mathbf{X}=\{X_{i,j}\in\mathcal{X}\}, (i,j)\in[m_n]\times [n]$ is a randomly generated ${m_n\times n}$ matrix with $X_{i,j}\overset{\text{i.i.d.}}{\sim} p_X$, where $p_X$ has a finite discrete support $\mathcal{X}=\{1,\dots,|\mathcal{X}|\}$.
\end{defn}

\begin{defn}{\textbf{(Column Repetition Pattern)}}\label{defn:repetitionpattern}
The \emph{column repetition pattern} $S^n=\{S_1,S_2,...,S_n\}$ is a random vector with $\displaystyle S_i\overset{\text{i.i.d.}}{\sim} p_S$, where the \emph{repetition distribution} $p_S$ has a finite integer support ${\{0,\dots,s_{\max}\}}$. Here $\delta\triangleq p_S(0)$ is called the \emph{deletion probability}.
\end{defn}

\begin{defn}{\textbf{(Anonymization Function)}}
    The \emph{anonymization function} $\sigma_n$ is a uniformly-drawn permutation of $[m_n]$.
\end{defn}

\begin{defn}{\textbf{(Obfuscation Distribution)}}
    The \emph{obfuscation} (also referred to as \emph{noise}) \emph{distribution} $p_{Y|X}$ is a conditional probability distribution with both $X$ and $Y$ taking values from $\mathcal{X}$.
\end{defn}

\begin{figure}[t]
\centerline{\includegraphics[width=0.75\textwidth,trim={0cm 25cm 0cm 0cm},clip]{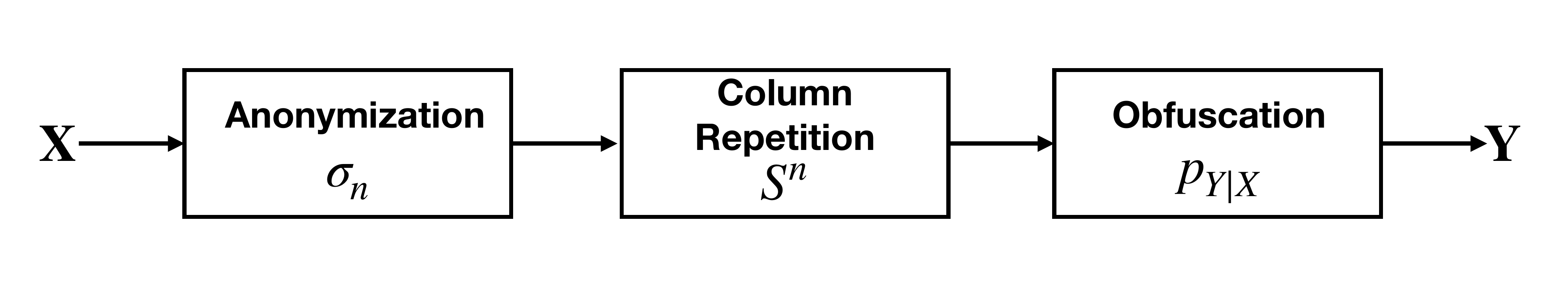}}
\caption{Relation between the anonymized database $\mathbf{X}$ and the labeled correlated database, $\mathbf{Y}$.}
\label{fig:systemmodel}
\end{figure}

\begin{defn}{\textbf{(Labeled Correlated Database)}}\label{defn:labeleddb}
Let $\mathbf{X}$ and $S^n$ be a mutually-independent ${(m_n,n,p_{X})}$ anonymized database and repetition pattern, $\sigma_n$ be an anonymization function, and $p_{Y|X}$ be an obfuscation distribution. Then, $\mathbf{Y}$ is called the \emph{labeled correlated database} if the $i$\textsuperscript{th} row $X^n_{i}$ of $\mathbf{X}$ and the $\sigma_n(i)$\textsuperscript{th} row $Y^{K_n}_{\sigma_n(i)}=[Y_{\sigma_n(i),1},\dots,Y_{\sigma_n(i),K_{n}}]$ of $\mathbf{Y}$ have the following relation      \begin{align}
    \Pr(Y^{K_{n}}_{\sigma_n(i)} = y^{K_n}|X_i^n=x^n)
    &= \prod\limits_{j:S_j\neq 0} \Pr(Y_{\sigma_n(i),K_{j-1}+1},\dots,Y_{\sigma_n(i),K_{j}}=y_{K_{j-1}+1},\dots,y_{K_{j}}|X_{i,j}=x_j)\label{eq:correlateddbident1}\\
    &= \prod\limits_{j:S_j\neq 0} \prod\limits_{s=1}^{S_{j}} p_{Y|X}(y_{K_{j-1}+s}|x_j),\label{eq:correlateddbident2}
\end{align}
where
\begin{align}
    K_{j}&\triangleq \sum\limits_{t=1}^j S_{t}.
\end{align}

Note that $S_j$ indicates the times the $j$\textsuperscript{th} column of $\mathbf{X}$ is repeated in $\mathbf{Y}$. When $S_j=0$, the $j$\textsuperscript{th} column of $\mathbf{X}$ is said to be \emph{deleted} and when $S_j>1$, the $j$\textsuperscript{th} column of $\mathbf{X}$ is said to be \emph{replicated}.

The $i$\textsuperscript{th} row $X_i$ of $\mathbf{X}$ and the $\sigma_n(i)$\textsuperscript{th} row $Y_{\sigma_n(i)}$ of $\mathbf{Y}$ are called \emph{matching rows}.
\end{defn}

The relationship between $\mathbf{X}$ and $\mathbf{Y}$, as described in Definition~\ref{defn:labeleddb}, is illustrated in Figure~\ref{fig:systemmodel}.

As often done in both the graph matching~\cite{shirani2021concentration} and the database matching~\cite{bakirtas2022seeded} literatures, we will assume the availability of a set of already-matched row pairs called \emph{seeds}, to be used in the detection of the underlying repetition pattern.

\begin{defn}{\textbf{(Seeds)}}
    Given a pair of anonymized and labeled correlated databases $(\mathbf{X},\mathbf{Y})$, a \emph{seed} is a correctly matched row pair with the same underlying repetition pattern $S^n$. A \emph{batch of $\Lambda_n$ seeds} is a pair of seed matrices of respective sizes $\Lambda_n\times n$ and $\Lambda_n\times \sum_{j=1}^n S_j$.
\end{defn}
For the sake of notational brevity and without loss of generality, we assume that the seed matrices $\mathbf{G}^{(1)}$ and $\mathbf{G}^{(2)}$ are not submatrices of $\mathbf{X}$ and $\mathbf{Y}$. Throughout, we will assume a seed size $\Lambda_n=\omega(\log n)=\omega(\log\log m_n)$ which is double-logarithmic with the number of users $m_n$.

Note that as the number of rows $m_n$ increases for a fixed number of columns $n$, so does the probability of mismatch due to the increased number of candidates. In turn, in de-anonymization (\emph{a.k.a.} matching) problems, the relationship between the row size $m_n$ and the column size $n$ directly impacts the de-anonymization performance. Hence, as done in~\cite{shirani8849392,bakirtas2021database,bakirtas2022database,bakirtas2022matching,bakirtas2022seeded}, we utilize the database growth rate, defined below, as the main performance metric.
\begin{defn}\label{defn:dbgrowthrate}{\textbf{(Database Growth Rate)}}
The \emph{database growth rate} $R$ of an $(m_n,n,p_X)$ anonymized database is defined as 
\begin{align}
    R&=\lim\limits_{n\to\infty} \frac{1}{n}\log m_n
\end{align}
\end{defn}

\begin{defn}{\textbf{(Distribution-Agnostically Achievable Database Growth Rate)}}\label{defn:achievable}
Consider a sequence of ${(m_n,n,p_X)}$ anonymized databases $\mathbf{X}$ with database growth rate $R$, an independent a repetition probability distribution $p_S$, an obfuscation distribution $p_{Y|X}$ and the resulting sequence of labeled correlated databases $\mathbf{Y}$. Given seed matrices $(\mathbf{G}^{(1)},\mathbf{G}^{(2)})$ with a seed size $\Lambda_n$, the database growth rate $R$ is said to be \emph{distribution-agnostically achievable} if there exists a successful matching scheme $\psi:(\mathbf{X},\mathbf{Y},\mathbf{G}^{(1)},\mathbf{G}^{(2)})\mapsto \hat{\sigma}_n$ \emph{that does not rely on any knowledge on $p_X$, $p_S$ and $p_{Y|X}$}, with
\begin{align}
\lim\limits_{n\to\infty}\Pr\left(\sigma_n(I)\neq\hat{\sigma}_n(I)\right)&= 0 \text{ where } I\sim\text{Unif}([m_n])\label{eq:errordefiniton}
\end{align}
where $\sigma_n$ is true the anonymization function.
\end{defn}

\begin{rem}
We note that the performance criterion described in Definition~\ref{defn:achievable} is known as \emph{almost-perfect recovery}.
Another well-known performance criterion is the \emph{perfect recovery} criterion where \emph{all} rows need to be matched. More formally, the perfect recovery criterion corresponds to
\begin{align}
    \lim\limits_{n\to\infty} \Pr(\sigma_n\neq \hat{\sigma}_n) = 0
\end{align}
Observe that the difference between the two performance criteria is that in the almost-perfect criterion allows a sublinear fraction of the database rows to be mismatched. This mismatch budget allows us to use tools such as typicality from information and communication theories. We refer to~\cite{kunisky2022strong} for an in-depth analytical comparison of the Gaussian database matching results under these different performance criteria.
\end{rem}

\begin{defn}{\textbf{(Distribution-Agnostic Matching Capacity)}}\label{defn:matchingcapacity}
     The \emph{distribution-agnostic matching capacity} $C$ is the supremum of the set of all distribution-agnostically achievable database growth rates corresponding to a database distribution $p_X$, repetition distribution $p_S$, obfuscation distribution $p_{Y|X}$, and seed size $\Lambda_n$.
\end{defn}

In this work, our main objective is the characterization of the \emph{distribution-agnostic} matching capacity $C$. We remark that this is in contrast with prior work (\cite{shirani8849392,bakirtas2021database,bakirtas2022seeded,bakirtas2022matching}) that assumed complete distributional knowledge.

Throughout, since we are interested in the distribution-agnostic matching capacity, we assume a positive database growth rate $R>0$. In other words, we assume $n\sim\log m_n$.

Our analysis hinges on the following assumptions:
\begin{rem}{\textbf{(Assumptions)}}\label{rem:assumptions}
\begin{enumerate}[label=\textbf{(\alph*)}]
    \item The anonymized database $\mathbf{X}$ and the column repetition pattern $S^n$ are known to be independent. It is also known that their components $\smash{X_{i,j}}$ and $\smash{S_i}$ are \emph{i.i.d.}, while the distributions $p_X$ and $p_S$ are not known.
    \item The conditional independence of the noisy replicas stated in \eqref{eq:correlateddbident2} is known, leading to a memoryless obfuscation model, whereas the noise distribution $p_{Y|X}$ is not.
    \item $|\mathcal{X}|$ and $s_{\max}$ are known.
\end{enumerate}
\end{rem}

As we argue in Sections~\ref{subsec:replicadetection} and \ref{subsec:deletiondetection}, $S^n$ can be detected without any assumptions on the correlation between $\mathbf{X}$ and $S^n$. In turn, one could test the independence of $\mathbf{X}$ and the estimate $\hat{S}^n$ of $S^n$. Furthermore, the \emph{i.i.d.} nature of the components of $\mathbf{X}$ and $\hat{S}^n$, and the obfuscation can be inferred via the Markov order estimation algorithm of~\cite{morvai2005order} with a probability of error vanishing in $n$, justifying \textbf{(a)}-\textbf{(b)}. Similarly, since $|\mathcal{X}|$ and $s_{\max}$ do not depend on $n$, they can easily be estimated with a vanishing probability of error, justifying \textbf{(c)}.

In our derivations, we rely on some well-known classical information-theoretic definitions and results which we present below for the sake of completeness.
\begin{defn}{\textbf{(Joint Entropy~\cite[Chapter 2.2]{cover2006elements})}}
    The joint (Shannon) entropy associated with a pair of discrete random variables $(X,Y)$ with a joint distribution $p_{X,Y}$ is defined as:
    \begin{align}
        H(X,Y) &\triangleq -\mathbb{E}[\log p_{X,Y}(X,Y)]
    \end{align}
\end{defn}

\begin{defn}{\textbf{(Joint Typicality~\cite[Chapter 7.6]{cover2006elements})}}
\label{defn:jointtypicality}
The $\epsilon$-typical set $A_\epsilon^{(n)}(X,Y)$ associated with discrete random variables $(X,Y)$ with joint distribution $p_{X,Y}$ is the set of all sequence pairs $(x^n,y^n)$ satisfying
\begin{align}
    \Big|-\frac{1}{n}&\log p_{X^n,Y^n}(x^n,y^n)-H(X,Y)\Big|\le \epsilon
\end{align}
where $H(\mathcal{X},\mathcal{Y})$ is the entropy rate of $(\mathcal{X},\mathcal{Y})$.
\end{defn}

\begin{prop}{\textbf{(Joint AEP~\cite[Theorem 7.6.1]{cover2006elements})}}
\label{prop:jointaep}
Let $\Tilde{X}^n$ and $\Tilde{Y}^n$ be generated according to the \emph{i.i.d.} marginal distributions $p_{X^n}$ and $p_{Y^n}$, independently. Then, the following holds:
    \begin{align}
        \Pr((\Tilde{X}^n,\Tilde{Y}^n)\in A_\epsilon^{(n)}(X,Y))&\le  2^{-n(I(X;Y)-3\epsilon)}
    \end{align}
    where $I(X;Y)\triangleq H(X)+H(Y)-H(X,Y)$ is the mutual information.   
\end{prop}

\newpage
\section{Main Result}
\label{sec:mainresult}
Our main result in Theorem~\ref{thm:mainresult} is on the distribution-agnostically achievable database growth rates when no prior information is provided on distributions $p_{X}$, $p_{Y|X}$, and $p_S$.
\begin{thm}{\textbf{(Distribution-Agnostic Matching Capacity)}}\label{thm:mainresult}
    Consider an anonymized and labeled correlated database pair. We assume that the underlying database distribution $p_X$, the obfuscation function $p_{Y|X}$, and the column repetition distribution $p_S$ are unknown. Given a seed size $\Lambda_n=\omega(\log n)$, the distribution-agnostic matching capacity is given by
    \begin{align}
        C&=I(X;Y^S|S) \label{eq:mainresult}
    \end{align}
    where $S\sim p_S$, $X\sim p_X$ and ${Y^S=Y_1,\dots,Y_S}$ such that $Y_i|X\overset{\text{i.i.d.}}{\sim} p_{Y|X}$.
\end{thm}

Theorem~\ref{thm:mainresult} implies that given a seed size ${\Lambda_n=\omega(\log n)=\omega(\log\log m_n)}$ we can perform matching as if we knew the underlying distributions $p_X$, $p_{Y|X}$ and $p_S$, and the actual column repetition pattern $S^n$ a priori. Hence comparing Theorem~\ref{thm:mainresult} with~\cite[Theorem 1]{bakirtas2023database}, we conclude that in the asymptotic regime where $n\to\infty$, not knowing the distributions and the realizations of the repetition pattern causes no loss in the matching capacity.

We stress that the proof of the converse part of Theorem~\ref{thm:mainresult} directly follows from ~\cite[Theorem~1]{bakirtas2023database} since distribution-agnostic matching capacity cannot be better than the one that assumes distributional information. The rest of this section is on the proof of the achievability part of Theorem~\ref{thm:mainresult}. In Section~\ref{subsec:replicadetection}, we present an algorithm for detecting noisy replicas and prove its asymptotic success. Then in Section~\ref{subsec:deletiondetection}, we propose a seeded deletion detection algorithm and derive a sufficient seed size that guarantees its asymptotic success. Subsequently, in Section~\ref{subsec:achievability}, we present our de-anonymization algorithm that incorporates Sections \ref{subsec:replicadetection} and \ref{subsec:deletiondetection}. Finally, in Section~\ref{subsec:noiseless}, we focus on the special case of no obfuscation, where seeds are obsolete. 
\subsection{Distribution-Agnostic Noisy Replica Detection}
\label{subsec:replicadetection}
Similar to~\cite{bakirtas2022seeded}, we use the running Hamming distances between the consecutive columns ${\displaystyle C_j^{(2)}}$ and $\displaystyle C_{j+1}^{(2)}$ of $\mathbf{Y}$, denoted by $W_j$, $j\in[K_n-1]$, where $K_n\triangleq\sum_{j=1}^n S_j$ as a permutation-invariant future of the labeled correlated database. More formally,
\begin{align}
    W_j & \triangleq \sum\limits_{t=1}^{m_n} \mathbbm{1}_{[Y_{t,j+1}\neq Y_{t,j}]},\hspace{2em} \forall j\in[K_n-1]\label{eq:RHD}
    \end{align}
We first note that 
\begin{align}
W_j&\sim \begin{cases}
        \text{Binom}(m_n,p_0),& \text{if }C^{(2)}_j \indep C^{(2)}_{j+1} \\
        \text{Binom}(m_n,p_1), & \text{otherwise}
    \end{cases}
    \label{eq:binomialmixture}
\end{align}
where 
\begin{align}
    p_0 &\triangleq 1-\sum\limits_{y\in\mathcal{X}} p_Y(y)^2 \label{eq:p0}\\
    p_1 &\triangleq 1-\sum\limits_{x\in\mathcal{X}} p_X(x) \sum\limits_{y\in\mathcal{X}} p_{Y|X}(y|x)^2\label{eq:p1}
\end{align}
From~\cite[Lemma 1]{bakirtas2022database}, we know that as long as the databases are dependent, \emph{i.e.,} $p_{X,Y}\neq p_X p_Y$, we have $p_0>p_1$ for any $p_{X,Y}$ suggesting that replicas can be potentially detected based on $W_j$ similar to~\cite{bakirtas2022seeded,bakirtas2022database}. However, the algorithm in~\cite{bakirtas2022seeded} relies on a threshold test with the threshold depending on $p_{X,Y}$ through $p_0$ and $p_1$. In Algorithm~\ref{alg:noisyreplicadetection}, we propose the following modification for the distribution-agnostic setting: We first construct the estimates $\hat{p}_0$ and $\hat{p}_1$ for the respective parameters $p_0$ and $p_1$ through the moment estimator proposed by Blischke~\cite{blischke1962moment} and then use the algorithm proposed in \cite[Section III-A]{bakirtas2022seeded} for replica detection. Note that we can use Blischke's estimator because the Binomial mixture is guaranteed to have two distinct components. More formally, the distribution of $W_j$ conditioned on $S^n$ is given by
    \begin{align}
    \Pr(W_j=w|S^n) &= \binom{m_n}{w} [\alpha p_0^{w} (1-p_0)^{m_n-w}+(1-\alpha) p_1^w (1-p_1)^{m_n-w}]
\end{align}
for $w=0,\dots,m_n$ where the mixing parameter $\alpha$ is given by
\begin{align}
    \alpha &= \frac{1}{K_n-1}\left(n-\sum\limits_{j=1}^n \mathbbm{1}_{[S_j=0]}\right)
\end{align}
Since $p_S$, and in turn $\delta$ are constant in $n$, it can easily be verified that as $n\to\infty$, $\alpha\overset{p}{\to} \frac{1-\delta}{\mathbb{E}[S]}$. Hence $\alpha$ is bounded away from both 0 and 1, suggesting that the moment estimator of~\cite{blischke1962moment} and in turn, Algorithm~\ref{alg:noisyreplicadetection} can be used to detect the replicas.

\begin{algorithm}[t]
\caption{Distribution-Agnostic Noisy Replica Detection Algorithm}\label{alg:noisyreplicadetection}
\Input{$(\mathbf{Y},m_n,K_n)$}
\Output{isReplica}
$W\gets $ RunningHammingDist($\mathbf{Y}$)\Comment*[r]{Eq.~\eqref{eq:RHD}}
$(\hat{p}_0,\hat{p}_1) \gets $EstimateParams$(W)$\Comment*[r]{Eq~\eqref{eq:paramest1}-\eqref{eq:paramest2}}
$\tau \gets \frac{\hat{p}_0+\hat{p}_1}{2}$\Comment*[r]{Threshold}
isReplica $\gets \varnothing$\;

\For{$j = 1$ \KwTo $K_n-1$}{
  \eIf{$W[j]\le m_n \tau$}{
    isReplica$[j] \gets$ TRUE\;
  }{
      isReplica$[j] \gets $ FALSE\;
    }
  
}
\end{algorithm}

\begin{lem}{\textbf{(Noisy Replica Detection)}}\label{lem: replica detection}
   Let $E_j$ denote the event that Algorithm~\ref{alg:noisyreplicadetection} fails to infer the correct relationship between $C^{(2)}_{j}$ and $C^{(2)}_{j+1}$, $j=1,\dots,K_n-1$. Then, given $m_n=\omega(\log n)$
\begin{align}
    \kappa^{(1)}_n&\triangleq \Pr(\bigcup\limits_{j=1}^{K-1} E_j)\to 0\text{ as }n\to\infty.
\end{align}
\end{lem}

\begin{proof}
See Appendix~\ref{proof: replica detection}.
\end{proof}

Note that the condition in Lemma~\ref{lem: replica detection} is automatically satisfied since $m_n$ is exponential in $n$ (Definition~\ref{defn:dbgrowthrate}). Furthermore, Algorithm~\ref{alg:noisyreplicadetection} has a runtime of $\mathcal{O}(m_n n)$, the computational bottleneck being the computation of the running Hamming distances $W_j$. Finally, we stress that as opposed to deletion detection, discussed in Section~\ref{subsec:deletiondetection}, no seeds are necessary for replica detection.
\subsection{Distribution-Agnostic Deletion Detection Using Seeds}
\label{subsec:deletiondetection}
In this section, we propose a deletion detection algorithm that utilizes the seeds. Since the replica detection algorithm of Section~\ref{subsec:replicadetection} (Algorithm~\ref{alg:noisyreplicadetection}) has a vanishing probability of error, for notational simplicity we will focus on a deletion-only setting throughout this subsection. Let $\mathbf{G}^{(1)}$ and $\mathbf{G}^{(2)}$ be the seed matrices with respective sizes $\Lambda_n\times n$ and $\Lambda_n\times \Tilde{K}_n$, and denote the $j$\textsuperscript{th} column of $\mathbf{G}^{(r)}$ with $\smash{G_j^{(r)}}$, $r=1,2$ where $\Tilde{K}_n\triangleq\sum_{j=1}^n \mathbbm{1}_{[S_j\neq 0]}$. Furthermore, for the sake of brevity, let $L_{i,j}$ denote the Hamming distance between $\smash{G^{(1)}_i}$ and $\smash{G^{(2)}_j}$ for $(i,j)\in[n]\times [\Tilde{K}_n]$. More formally, let 
\begin{align}
    L_{i,j}&\triangleq \sum\limits_{t=1}^{\Lambda_n} \mathbbm{1}_{[G^{(1)}_{t,i}\neq G^{(2)}_{t,j}]}\label{eq:crosshammingdist}
\end{align}
Observe that 
\begin{align}
L_{i,j}&\sim \begin{cases}
        \text{Binom}(\Lambda_n,q_0),& G^{(1)}_i \indep G^{(2)}_{j} \\
        \text{Binom}(\Lambda_n,q_1), & \text{otherwise}
    \end{cases}
\end{align}where
\begin{align}
    q_0 &= 1-\sum\limits_{x\in\mathcal{X}} p_X(x) p_Y(x)\\
    q_1 &= 1-\sum\limits_{x\in\mathcal{X}} p_{X,Y}(x,x)
\end{align}

Thus, we have a problem seemingly similar to the one in Section~\ref{subsec:replicadetection}. However, we cannot utilize similar tools here because of the following: 
\begin{enumerate}[label = \textbf{(\roman*)}]
    \item Recall that the two components $p_0$ and $p_1$ of the Binomial mixture discussed in Section~\ref{subsec:replicadetection} were distinct for any underlying joint distribution $p_{X,Y}$ as long as the databases are dependent, \emph{i.e.,} $p_{X,Y}\neq p_X p_Y$. Unfortunately, the same idea does not automatically work here as demonstrated by the following example: Suppose $X_{i,j}\sim\text{Unif}(\mathcal{X})$, and the transition matrix $\mathbf{P}$ associated with $p_{Y|X}$ has unit trace. Then, 
\begin{align}
    q_0-q_1&=\sum\limits_{x\in\mathcal{X}} p_{X,Y}(x,x) - p_X(x) p_Y(x)\\
    &= \frac{1}{|\mathcal{X}|} \sum\limits_{x\in\mathcal{X}} p_{Y|X}(x|x) - p_Y(x)\\
    &= \frac{1}{|\mathcal{X}|}(\text{tr}(\mathbf{P})-1)\\
    &= 0
\end{align}
In~\cite{bakirtas2022seeded}, this problem was solved using the following modification: Based on $p_{X,Y}$, a bijective remapping $\Phi\in \mathcal{S}(\mathcal{X})$ is picked and applied to all the entries of $\mathbf{G}^{(2)}$ to obtain $\mathbf{G}^{(2)}(\Phi)$ before computing the Hamming distances $L_{i,j}$, where $\mathcal{S}(\mathcal{X})$ denotes the symmetry group of $\mathcal{X}$. Denoting the resulting version of the Hamming distance by $L_{i,j}(\Phi)$, where
\begin{align}
    L_{i,j(\Phi)}&\triangleq \sum\limits_{t=1}^{\Lambda_n} \mathbbm{1}_{[G^{(1)}_{t,i}\neq G^{(2)}_{t,j}(\Phi)]}\label{eq:crosshammingdistremapped}
\end{align}
it was proved in~\cite[Lemma 2]{bakirtas2022seeded} that there as long as $p_{X,Y}\neq p_X p_Y$, there exists $\Phi\in \mathcal{S}(\mathcal{X})$ such that $q_0(\Phi)\neq q_1(\Phi)$. We will call such $\Phi$ a \emph{useful remapping}.

\item In the known distribution setting, we chose a useful remapping $\Phi$ and threshold $\tau_n$ for Hamming distances based on $p_{X,Y}$. In Section~\ref{subsec:replicadetection}, we solved the distribution-agnostic case via parameter estimation in Binomial mixtures. However, the same approach does not work here. Suppose the $j$\textsuperscript{th} retained column $\smash{G_{j}^{(2)}}$ of $\mathbf{G}^{(2)}$ is correlated with $\smash{G_{r_j}^{(1)}}$. Then the $j$\textsuperscript{th} column of $\mathbf{L}(\Phi)$ will have a Binom($\Lambda_n,q_1(\Phi)$) component in the $r_j$\textsuperscript{th} row, whereas the remaining $n-1$ rows will contain Binom($\Lambda_n,q_0(\Phi)$) components, as illustrated in Figure~\ref{fig:deletiondetection}. 
Hence, it can be seen that the mixture parameter $\beta$ of this Binomial mixture distribution approaches 1 since
\begin{align}
    \beta &= \frac{(n-1)\Tilde{K}_n}{n \Tilde{K}_n} = 1-\frac{1}{n}
\end{align}
This imbalance prevents us from performing a parameter estimation as done in Algorithm~\ref{alg:noisyreplicadetection}.
\end{enumerate}

\begin{figure}[t]
\centerline{\includegraphics[width=0.65\textwidth]{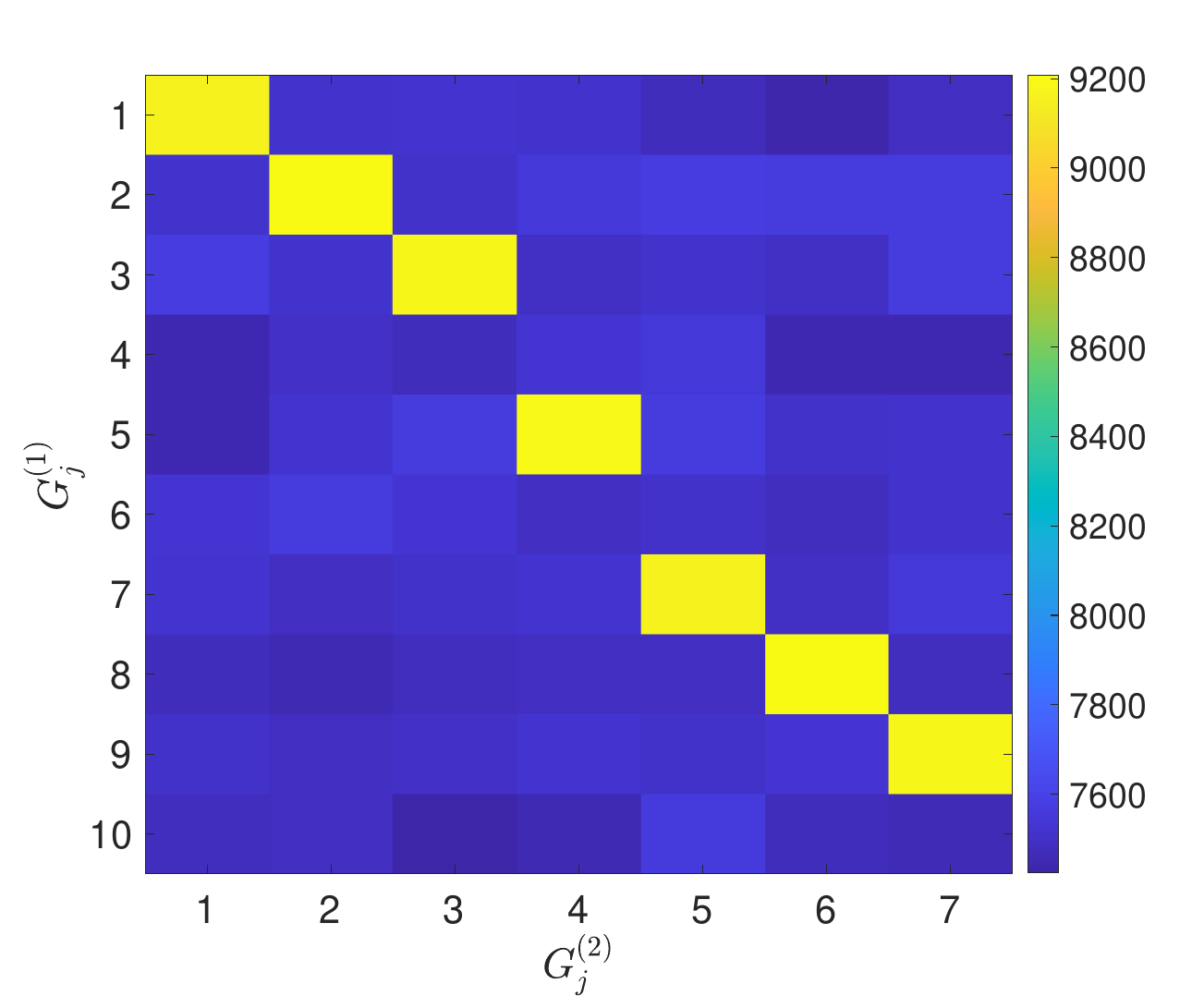}}
\caption{Hamming distances between the columns of $\mathbf{G}^{(1)}$ and $\mathbf{G}^{(2)}$ with $n=10$, $\Tilde{K}_n=7$ and $\Lambda_n=10^4$ for $q_0\approx 0.76$ and $q_1\approx 0.92$. The $(i,j)$\textsuperscript{th} element corresponds to $L_{i,j}$, with the color bar indicating the approximate values. It can be seen that there are no outliers in the $4$\textsuperscript{th}, $6$\textsuperscript{th}, and $10$\textsuperscript{th} rows. Hence, it can be inferred that $I_{\text{del}}=(4,6,10)$.}
\label{fig:deletiondetection}
\end{figure}

In Algorithm~\ref{alg:deletiondetection}, we exploit the aforementioned observation that for a useful mapping $\Phi$, in each column of $\mathbf{L}(\Phi)$, there is exactly one element with a different underlying distribution, while the remaining $n-1$ entries are \emph{i.i.d.}, rendering this entry an \emph{outlier}. Note that $L_{i,j}(\Phi)$ being an outlier corresponds to $G_i^{(1)}$ and $G_j^{(2)}(\Phi)$ being dependent, and in turn $S_i\neq 0$. On the other hand, the lack of outliers in any given column of $\mathbf{L}(\Phi)$ implies that $\Phi$ is not useful. Thus, Algorithm~\ref{alg:deletiondetection} is capable of deciding whether a given remapping is useful or not. In fact, the algorithm sweeps over all elements of $\mathcal{S}(\mathcal{X})$ until we encounter a useful one.

To detect the outliers in $\mathbf{L}(\Phi)$, we propose to use the absolute deviations $\mathbf{M}(\Phi)$, \emph{i.e.,} distances of $L_{i,j}(\Phi)$ to the sample mean $\mu(\Phi)$ of $\mathbf{L}(\Phi)$. More formally, we have
\begin{align}
    \mu(\Phi) &\triangleq \frac{1}{n \Tilde{K}_n}\sum\limits_{i=1}^{n} \sum\limits_{j=1}^{\Tilde{K}_n} L_{i,j}(\Phi)\label{eq:samplemean}\\
    M_{i,j}(\Phi) &\triangleq |L_{i,j}(\Phi)-\mu(\Phi)|,\hspace{1em} \forall (i,j)\in [n]\times[\Tilde{K}_n]\label{eq:absdeviations}
\end{align}
In Algorithm~\ref{alg:deletiondetection}, we test these absolute deviations $M_{i,j}(\Phi)$ against a threshold $\hat{\tau}_n$ independent of the underlying distribution $p_{X,Y}$. If $M_{i,j}(\Phi)$ is lower than $\hat{\tau}_n$, we detect retention \emph{i.e.,} non-deletion.

Note that this step is equivalent to utilizing
Z-scores (also known as standard scores), a well-studied concept in statistical outlier detection~\cite{moore2007basic}, where the absolute deviations are also divided by the sample standard deviation. In Algorithm~\ref{alg:deletiondetection}, for the sake of notational brevity, we will avoid such division. 

We observe that the runtime of Algorithm~\ref{alg:deletiondetection} is $\mathcal{O}(\Lambda_n n^2)$ due to the computation of $\mathbf{L}(\Phi)$.

\begin{algorithm}[H]
\caption{Distribution-Agnostic Seeded Deletion Detection Algorithm}\label{alg:deletiondetection}
\Input{$(\mathbf{G}^{(1)},\mathbf{G}^{(2)},\Lambda_n,n, \Tilde{K}_n,\mathcal{X})$}
\Output{$\hat{I}_R$}
$\mathcal{S}(\mathcal{X})\gets$ SymmetryGroup($\mathcal{X}$)\;

$\hat{\tau}_n\gets 2 \Lambda_n^{2/3} (\log n)^{1/3}$\Comment*[r]{Threshold}

\For{$s \gets 1$ \KwTo $|\mathcal{X}|!$}{
$\hat{I}_R \gets\varnothing$\;
$\Phi\gets\mathcal{S}(\mathcal{X})[s]$\Comment*[r]{Pick a remapping.} 
$\mathbf{L}(\Phi) \gets $ HammDist($\mathbf{G}^{(1)},\mathbf{G}^{(2)}(\Phi)$)\Comment*[r]{Eq. \eqref{eq:crosshammingdist}}
$\mu(\Phi)\gets$ SampleMean($\mathbf{L}(\Phi) $)\Comment*[r]{Eq.~\eqref{eq:samplemean}}
$\mathbf{M}(\Phi) \gets |\mathbf{L}(\Phi)-\mu(\Phi)|$\;

\For{$j\gets1$ \KwTo $\Tilde{K}_n$}{
count $\gets0$\; 
\For{$i\gets1$ \KwTo $n$}{
\If{$\mathbf{M}(\Phi)[i][j]\le\hat{\tau}_n$}{
$\hat{I}_R\gets \hat{I}_R\cup \{i\} $\;
count $\gets$ count $+$ $1$\;
}

}
\blue{/* count = 0: no outliers ($\Phi$ is useless). */}\\
\blue{/* count $>$ 1: misdetection. */}\\
\eIf{$\textup{count}>1$}{\Return ERROR}{\If{$\textup{count}=0$}{Skip to next $\Phi$\;}}
}
\Return $\hat{I}_R$\;
}

\end{algorithm}
\newpage

Lemma~\ref{lem: deletion detection} below states that for sufficient seed size, ${\Lambda_n=\omega(\log n)= \omega(\log\log m_n)}$, Algorithm~\ref{alg:deletiondetection} works correctly with high probability.
\begin{lem}{\textbf{(Deletion Detection)}}
\label{lem: deletion detection}
    Let ${I_R=\{j\in[n]: S_j\neq 0\}}$ be the true set of indices of retained columns and $\hat{I}_R$ be its estimate output by Algorithm~\ref{alg:deletiondetection}. Then for any seed size ${\Lambda_n=\omega(\log n)}$, we have
    \begin{align}
        \kappa^{(2)}_n\triangleq \lim\limits_{n\to\infty} \Pr(\hat{I}_R = I_R) &= 1
    \end{align}
\end{lem}

\begin{proof}
See Appendix~\ref{proof: deletion detection}.
\end{proof}

\subsection{Distribution-Agnostic De-Anonymization Scheme}
\label{subsec:achievability}
In this section, we propose a de-anonymization scheme by combining Algorithm~\ref{alg:noisyreplicadetection} and Algorithm~\ref{alg:deletiondetection}, and performing a modified version of the typicality-based de-anonymization proposed in~\cite{bakirtas2022seeded}. This then leads to the achievability proof of Theorem~\ref{thm:mainresult} in Section~\ref{subsec:proofofmainresult}. 

Given the database pair $(\mathbf{X},\mathbf{Y})$ and the corresponding seed matrices $(\mathbf{G}^{(1)},\mathbf{G}^{(2)})$, the proposed de-anonymization scheme given in Algorithm~\ref{alg:matchingscheme} is as follows: 
\begin{enumerate}[label=\textbf{ (\roman*)}]
    \item Detect the replicas through Algorithm~\ref{alg:noisyreplicadetection}.
    \item Remove all the extra replica columns from the seed matrix $\mathbf{G}^{(2)}$ to obtain $\tilde{\mathbf{G}}^{(2)}$ and perform seeded deletion detection via Algorithm~\ref{alg:deletiondetection} using $\mathbf{G}^{(1)},\tilde{\mathbf{G}}^{(2)}$. At this step, we have an estimate $\hat{S}^n$ of the column repetition pattern $S^n$.
    \item Based on $\hat{S}^n$ and the matching entries in $\mathbf{G}^{(1)},\tilde{\mathbf{G}}^{(2)}$, obtain the maximum likelihood estimate~\cite{anderson1957statistical} $\hat{p}_{X,Y^S|S}$ of $p_{X,Y^S|S}$
    where
    \begin{align}
     \hat{p}_X(x)&\triangleq \frac{1}{\Lambda_n n} \sum\limits_{i=1}^{\Lambda_n}\sum\limits_{j=1}^{n}\mathbbm{1}_{[G_{i,j}^{(1)}=x]},\hspace{2em}\forall x\in\mathcal{X}\label{eq:estimateddist1}\\
     \hat{p}_{Y|X}(y|x) &\triangleq\frac{\sum\limits_{i=1}^{\Lambda_n}\sum\limits_{j=1}^{\tilde{K}_n}\mathbbm{1}_{[G^{(1)}_{i,r_j}=x,\tilde{G}^{(2)}_{i,j}=y ]}}{\sum\limits_{i=1}^{\Lambda_n}\sum\limits_{j=1}^{\tilde{K}_n}\mathbbm{1}_{[\tilde{G}^{(2)}_{i,j}=y ]}} ,\hspace{1em}\forall (x,y)\in\mathcal{X}^2\\
     \hat{p}_S(s) &\triangleq\frac{1}{n}\sum\limits_{j=1}^n \mathbbm{1}_{[S_j=s]},\hspace{5.5em}\forall s\ge0\label{eq:estimateddist3}
    \end{align}
    and construct
    \begin{align}
        \hat{p}_{X,{Y}^S|S}(x,y^s|s)&=\begin{cases}
      \hat{p}_X(x) \mathbbm{1}_{[y^s = \ast]} &\text{if } s=0\\
      \hat{p}_X(x) \prod\limits_{j=1}^s \hat{p}_{Y|X}(y_j|x) &\text{if } s\ge 1
    \end{cases}
    \end{align}
with $y^s=y_1\dots y_s$ and $\ast$ denoting erasure.
    \item Using $\hat{S}^n$, place markers between the noisy replica runs of different columns to obtain $\tilde{\mathbf{Y}}$. If a run has length 0, \emph{i.e.} deleted, introduce a column consisting of erasure symbol $\ast\notin\mathcal{X}$.
    \item Fix $\epsilon>0$. Match the $l$\textsuperscript{th} row $Y^{K_n}_l$ of $\tilde{\mathbf{Y}}$ with the $i$\textsuperscript{th} row $X^n_i$ of {$\mathbf{X}$}, if $X_i$ is the only row of {$\mathbf{X}$} jointly $\epsilon$-typical (See Definition~\ref{defn:jointtypicality}) with $Y^{K_n}_l$ according to $\hat{p}_{X,Y^S|S}$, assigning $\hat\sigma_n(i)=l$.
Otherwise, declare an error.
\end{enumerate}
Finally, note that the runtime of Algorithm~\ref{alg:matchingscheme} is $\mathcal{O}(m_n^2 n)$ due to the typicality check (each $\mathcal{O}(n)$) for all row pairs $(X_i^n,Y_j^{K_n})$, $(i,j)\in[m_n]^2$.

\subsection{Proof of Theorem~\ref{thm:mainresult}}
\label{subsec:proofofmainresult}
We are now ready to prove Theorem~\ref{thm:mainresult} where we use Algorithm~\ref{alg:matchingscheme} to obtain a matching scheme $\hat{\sigma}_n$.

\begin{algorithm}[t]
\caption{Distribution-Agnostic Database De-Anonymization Scheme}\label{alg:matchingscheme}
\Input{$(\mathbf{X},\mathbf{Y},\epsilon,\mathbf{G}^{(1)},\mathbf{G}^{(2)})$}
\Output{$\hat{\sigma}_n$}
\textup{isReplica}$\gets$ Alg.\ref{alg:noisyreplicadetection}($\mathbf{Y}$)\Comment*[r]{\textbf{Step (i)}.}
\textup{isDeleted}$\gets$ Alg.\ref{alg:deletiondetection}($\mathbf{G}^{(1)},\mathbf{G}^{(2)}$,\textup{isReplica})\Comment*[r]{\textbf{Step (ii)}.}
$\hat{S}^n \gets $ EstimateRepetitionPattern(isReplica,isDeleted)\;
$\hat{p}_{X,Y^S|S}\gets $ EstimateDist($\mathbf{G}^{(1)},\mathbf{G}^{(2)},\hat{S}^n$)\Comment*[r]{\textbf{Step (iii)}.}
$\tilde{\mathbf{Y}}\gets$ MarkerAddition($\mathbf{Y},\hat{S}^n$)\Comment*[r]{\textbf{Step (iv)}.}

\For{$i = 1$ \KwTo  \textup{rowSize($\mathbf{X}$)}}{
\textup{count}$\gets 0$\;
\For{$j = 1$ \KwTo  \textup{rowSize($\tilde{\mathbf{Y}}$)}}{
  \If{\textup{isJointTypical(}$\mathbf{X}[i][:],\mathbf{Y}[j][:], \hat{S}^n,\hat{p}_{X,Y^S|S},\epsilon$\textup{)}}{
    $\hat{\sigma}_n[i] \gets$ j\;
    count$\gets$ count + 1\;
  }
  
}
\If{\textup{count} $\neq 1$}{
$\hat{\sigma}_n[i] \gets 0$\Comment*[r]{Matching error.}
}
}

\end{algorithm}

Let $\kappa_n^{(1)}$ and $\kappa_n^{(2)}$ be the error probabilities of the noisy replica detection (Algorithm~\ref{alg:noisyreplicadetection}) and the seeded deletion (Algorithm~\ref{alg:deletiondetection}) algorithms, respectively.
Using~\eqref{eq:estimateddist1}-\eqref{eq:estimateddist3} and the Law of Large Numbers, we have 
\begin{align}
    \hat{p}_{X,Y^S|S}\overset{p}{\to} p_{X,Y^S|S}
\end{align}
and by the Continuous Mapping Theorem~\cite[Theorem 2.3]{van2000asymptotic} we have
\begin{align}
    \hat{H}(X,Y^S|S)&\overset{p}{\to} H(X,Y^S|S)\\
I(\hat{X};\hat{Y}^{\hat{S}}|\hat{S})& \overset{p}{\to} I(X;Y^S|S)
\end{align}
where $\hat{H}(X,Y^S|S)$ and $\hat{I}(X,Y^S|S)$ denote the conditional joint entropy and conditional mutual information associated with $\hat{p}_{X,Y^S|S}$, respectively. Thus, for any $\epsilon>0$ we have
\begin{align}
\kappa_n^{(3)}&\triangleq\Pr(|\hat{H}(X,Y^S|S)-H(X,Y^S|S)|>\epsilon)\overset{n\to\infty}{\longrightarrow}0\\
\kappa_n^{(4)}&\triangleq\Pr(|\hat{I}(X,Y^S|S)-I(X,Y^S|S)|>\epsilon)\overset{n\to\infty}{\longrightarrow}0
\end{align}

Using Proposition~\ref{prop:jointaep} and a series of union bounds and triangle inequalities involving $\hat{H}(X,Y^S|S)$, $H(X,Y^S|S)$, $\hat{I}(X,Y^S|S)$, and $I(X,Y^S|S)$, the probability of error of the de-anonymization scheme (See~\eqref{eq:errordefiniton}.) can be bounded as 
\begin{align}
    \Pr(\text{Error})&\le 2^{-n(I(X;Y^S|S)-4 \epsilon-R)}+\epsilon + \sum\limits_{i=1}^4 \kappa_n^{(i)}\label{eq:error}\\
    &\le 2\epsilon
\end{align}
as $n\to\infty$ as long as $R<I(X;Y^S|S)-4\epsilon$, concluding the proof of the main result.

\subsection{No-Obfuscation Setting}
\label{subsec:noiseless}
In Sections~\ref{subsec:replicadetection} and \ref{subsec:deletiondetection}, we focused on the noisy/obfuscated setting and proposed algorithms for detecting noisy replicas (Algorithm~\ref{alg:noisyreplicadetection}) and column deletions (Algorithm~\ref{alg:deletiondetection}). As discussed in Sections~\ref{subsec:replicadetection} and \ref{subsec:deletiondetection}, the key idea behind detection is either extracting permutation-invariant features of the columns (Algorithm~\ref{alg:noisyreplicadetection}) or assuming the correct matching is given for the seeds (Algorithm~\ref{alg:deletiondetection}).
We showed that in the general noisy setting, for the latter approach to succeed a double-logarithmic seed size $\Lambda_n$ is sufficient.

Now, we will focus on the no-obfuscation (repetition-only) setting, where 
\begin{align}
    p_{Y|X}(y|x) &= \mathbbm{1}_{[y=x]}, \hspace{1em} \forall (x,y)\in \mathcal{X}^2
\end{align}
as a special case. Following a similar approach to Algorithm~\ref{alg:noisyreplicadetection}, through the extraction of a new set of permutation-invariant features of the database columns, we argue that seeds are not required for repetition (both \emph{replica and deletion}) detection. Specifically, we will replace Algorithm~\ref{alg:noisyreplicadetection} and Algorithm~\ref{alg:deletiondetection} a single the histogram-based detection algorithm of~\cite{bakirtas2022database} which works as follows:
\begin{enumerate}[label = \textbf{(\roman*)}]
    \item First, we construct the respective column histogram matrices $\mathbf{H}^{(1)}$ and $\mathbf{H}^{(2)}$ and  of $\mathbf{X}$ and $\mathbf{Y}$ as follows:
    \begin{align}
        H^{(1)}_{i,j} &= \sum\limits_{t=1}^{m_n} \mathbbm{1}_{[X_{t,j}=i]},\hspace{1em} \forall i\in\mathcal{X}, \forall j\in[n]\label{eq:hist1}\\
        H^{(2)}_{i,j} &= \sum\limits_{t=1}^{m_n} \mathbbm{1}_{[Y_{t,j}=i]},\hspace{1em} \forall i\in\mathcal{X}, \forall j\in[K_n]\label{eq:hist2}
    \end{align}

    \item Next, we count the number of times the $j$\textsuperscript{th} column $\smash{{H}^{(1)}_j}$ of $\mathbf{H}^{(1)}$ is present in $\mathbf{M}^{(2)}$. If $\smash{{H}^{(1)}_j}$ is present $s$ times in $\mathbf{H}^{(2)}$, we assign $\hat{S}_j = s$, where $\hat{S}^n$ is the estimate of the repetition pattern $S^n$. If $\smash{{H}^{(1)}_j}$ is absent in $\mathbf{H}^{(2)}$, we assign $\hat{S}_j = 0$.
\end{enumerate}

Note that a repetition error could occur only when there are identical columns in $\mathbf{H}^{(1)}$ whose probability, we argue, vanishes with the column size $n$ in the following lemma:
\begin{lem}{\textbf{(Asymptotic Uniqueness of the Column Histograms \cite[Lemma 1]{bakirtas2022database})}}\label{lem:histogram}
Let ${H}^{(1)}_j$ denote the histogram of the $j$\textsuperscript{th} column of $\mathbf{X}$.
Then, for any $m_n=\omega(n^\frac{4}{|\mathcal{X}|-1})$
\begin{align}
    \lim\limits_{n\to \infty}\Pr\left(\exists i,j\in [n],\: i\neq j,H^{(1)}_i=H^{(1)}_j\right)\to 0. \label{eq:HBD}
\end{align}
\end{lem}
Note that since $m_n$ is exponential in $n$ (See Definition~\ref{defn:dbgrowthrate}), $m_n=\omega(n^s)$ $\forall s\in\mathbb{Q}$. Hence, the order relation given in Lemma~\ref{lem:histogram} is automatically satisfied and hence repetition detection could be performed without any seeds in the no-obfuscation setting with a runtime of $\mathcal{O}(m_n n)$ due to the computation of the histogram matrices $(\mathbf{H}^{(1)},\mathbf{H}^{(2)})$.

For the non-obfuscation setting, we also consider a modified de-anonymization scheme, given in Algorithm~\ref{alg:noiselessmatching}, that does not rely on any estimates of the underlying distributions. Specifically, we perform de-anonymization via exact sequence matching as follows:
\begin{enumerate}[label = \textbf{(\roman*)}]
    \item First, we perform repetition detection via the histogram-based detection algorithm described above.
    \item Next, we discard the deleted columns from $\mathbf{X}$ to obtain $\smash{\bar{\mathbf{X}}}$. 
    \item Similarly, we discard all additional copies in a repetition run in $\mathbf{Y}$ to obtain $\bar{\mathbf{Y}}$.
    \item Finally, we match the $i$\textsuperscript{th} row $X_i^n$ with the $j$\textsuperscript{th} row $Y_j^{K_n}$, assigning $\hat{\sigma}_n(i)=j$, if $\smash{\bar{Y}_j^{\hat{K}_n}}$ is the only row of $\bar{\mathbf{Y}}$ equal to $\smash{\bar{X}_i^{\hat{K}_n}}$.
\end{enumerate}

The histogram-based detection and the modified de-anonymization scheme are given in Algorithm~\ref{alg:noiselessmatching}. We note that the overall runtime of Algorithm~\ref{alg:noiselessmatching} is $\mathcal{O}(m_n^2 n)$ due to the exact row matching (each $\mathcal{O}(n)$) of $m_n^2$ row pairs.

The matching scheme described above and given in Algorithm~\ref{alg:noiselessmatching} leads to the achievability result given in Theorem~\ref{thm: noiseless}. The converse immediately follows from the distribution-aware matching capacity result of~\cite[Theorem 2]{bakirtas2023database}. 
\begin{thm}{\textbf{(Distribution-Agnostic Matching Capacity in the No-Obfuscation Setting)}} \label{thm: noiseless}
    Consider an anonymized and labeled correlated database pair. We assume that the underlying database distribution $p_X$ and the column repetition distribution $p_S$ are unknown. Suppose there is no obfuscation, \emph{i.e.,}
    \begin{align}
    p_{Y|X}(y|x) &= \mathbbm{1}_{[y=x]}, \hspace{1em} \forall (x,y)\in \mathcal{X}^2.
\end{align}
    Then, for any seed size $\Lambda_n$, the distribution-agnostic matching capacity $C$ is given by
    \begin{align}
        C&= (1-\delta) H(X).\label{eq:noiselessachievability}
    \end{align}
\end{thm}
\begin{proof}
    See Appendix~\ref{proof: noiseless}.
\end{proof}
Observe that Theorem~\ref{thm: noiseless} states that the repetition distribution $p_S$ appears in the matching capacity only through the deletion probability $\delta$ in the no-obfuscation setting.

\begin{algorithm}[H]
\caption{Distribution-Agnostic De-Anonymization In No-Obfuscation Setting}\label{alg:noiselessmatching}
\Input{$(\mathbf{X},\mathbf{Y})$}
\Output{$\hat{\sigma}_n$}
$(\mathbf{H}^{(1)},\mathbf{H}^{(2)})\gets$ ColHist($\mathbf{X},\mathbf{Y}$)\Comment*[r]{Eq.~\eqref{eq:hist1}-\eqref{eq:hist2}.}
\blue{/* Histogram-based repetition detection */}\\
\For{$i=1$ \KwTo \textup{columnSize(${\mathbf{H}}^{(1)}$)}}{
count$\gets 0$\;
\For{$j=1$ \KwTo \textup{columnSize(${\mathbf{H}}^{(2)}$)}}{
\If{${\mathbf{H}}^{(2)}[:][j]={\mathbf{H}}^{(1)}[:][i]$}{
count$\gets$ count + 1\;
}
}
$\hat{S}[i]\gets $ count\;
}
$(\bar{\mathbf{X}},\bar{\mathbf{Y}})\gets \textup{ColumnDiscard}(\mathbf{X},\mathbf{Y},\hat{S})$

\blue{/* Exact sequence matching */}\\
\For{$i = 1$ \KwTo  \textup{rowSize($\mathbf{X}$)}}{
\textup{count}$\gets 0$\;
\For{$j = 1$ \KwTo  \textup{rowSize($\mathbf{Y}$)}}{
  \If{$\bar{\mathbf{X}}[i][:]=\bar{\mathbf{Y}}[j][:]$}{
    $\hat{\sigma}_n[i] \gets$ j\;
    count$\gets$ count + 1\;
  } 
}
\If{\textup{count} $\neq 1$}{
$\hat{\sigma}_n[i] \gets 0$\Comment*[r]{Matching error.}
}
}
\end{algorithm}

\newpage

\section{Non-Asymptotic Regime}
\label{sec:experiments}
In Section~\ref{sec:mainresult}, we devised various detection and de-anonymization algorithms and proved their performances (in terms of error probabilities and matching capacity) in the asymptotic regime where the column size $n$ grows to infinity. In this section, through extensive experiments, we will study the performances of these algorithms when $n$ is finite.

Throughout this section, unless stated otherwise, we will focus on the following simulation configuration:
\begin{itemize}
    \item \emph{Database size $(m_n,n)$:} We choose column size $n=100$ with the exception of the evaluation of the de-anonymization scheme in Section~\ref{subsec:typicalityexp} where we take $n=25$. We evaluate the performance of our algorithms for a range of values of row sizes $m_n$.
    \item \emph{Database distribution $p_X$:} Without loss of generality, we focus on the uniform distribution $p_X(x)=\frac{1}{|\mathcal{X}|}, \forall x\in\mathcal{X}$ where we arbitrarily take the alphabet size $|\mathcal{X}|=5$. In Section~\ref{subsec:HBDexp}, we will sweep over $|\mathcal{X}|$.
    \item \emph{Obfuscation $p_{Y|X}$:} We will consider the $|\mathcal{X}|$-ary symmetric channel  with a crossover probability $\epsilon$. More formally, for any $(x,y)\in\mathcal{X}^2$, we have
    \begin{align}
        p_{Y|X}(y|x) &=\begin{cases}
            1-\epsilon,&\text{if }y=x\\
            \frac{\epsilon}{|\mathcal{X}|-1},& \text{if }y\neq x
        \end{cases}
    \end{align}
    where we simulate with a range of values of $\epsilon$.
    \item \emph{Repetition distribution $p_S$:} We consider a \say{deletion-duplication} model for the repetition distribution with
    \begin{align}
        p_S(s) &= \begin{cases}
            \delta,&\text{if }s=0\\
            1-\delta-\gamma,& \text{if }s=1\\
            \gamma,& \text{if }s=2
        \end{cases}
    \end{align}
    where $\delta = 0.3$ and $\gamma =0.2$.
\end{itemize}
\subsection{Noisy Replica Detection}
\label{subsec:NRDexp}
Figure~\ref{fig:NRD} demonstrates the relationship between the replica detection error $\kappa^{(1)}_n$ of Algorithm~\ref{alg:noisyreplicadetection} and the row size $m_n$. Note that the linear relation in the semi-logarithmic plot validates the exponential decay of $\kappa^{(1)}_n$ with $m_n$ stated in~\eqref{eq:replicadetectionlast}.

\begin{figure}[t]
\centerline{\includegraphics[width=0.65\textwidth,trim={0cm 6cm 1cm 6cm},clip]{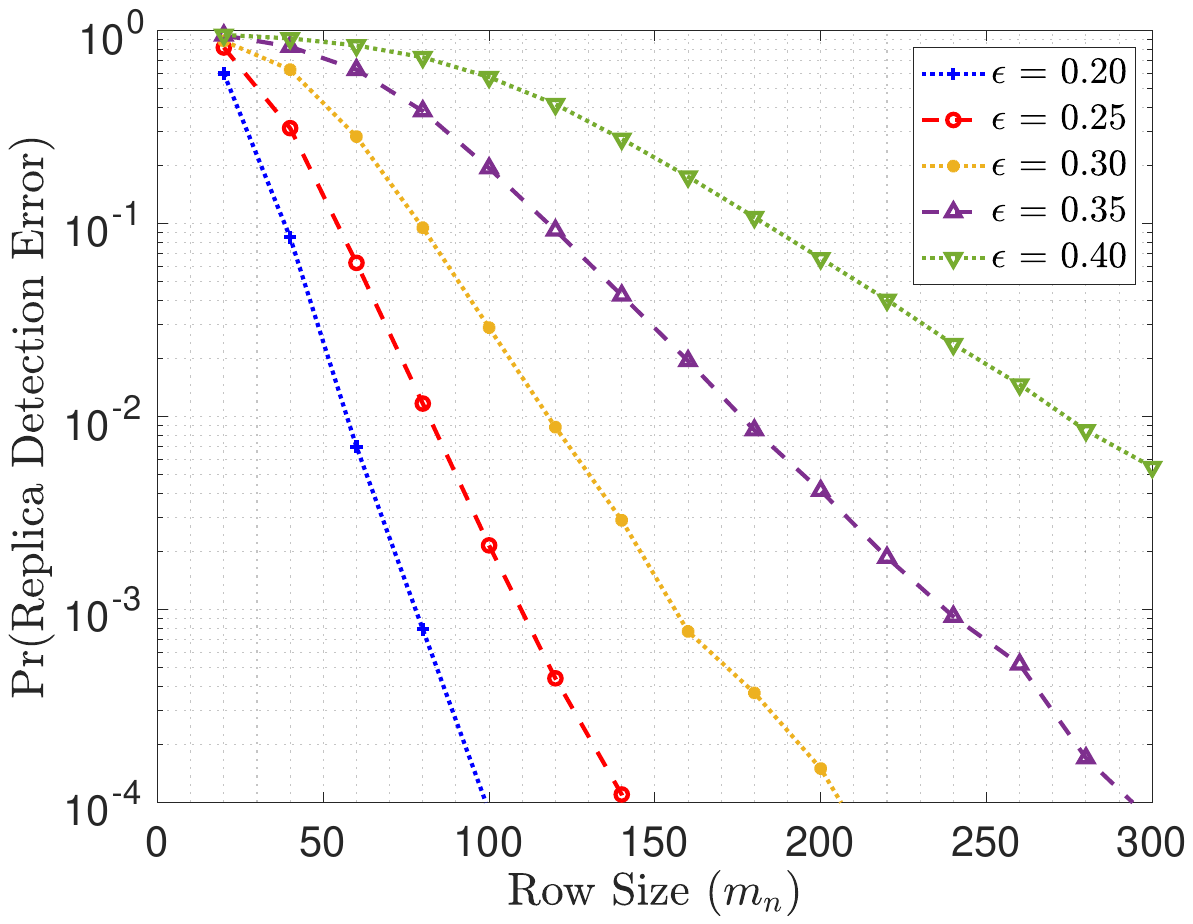}}
\caption{Probability of error of the noisy replica detection algorithm (Algorithm~\ref{alg:noisyreplicadetection}) $\kappa^{(1)}_n$ vs. the row size $m_n$ with $10^5$ trials. The y-axis is given in logarithmic scale to validate the exponential relation between the error probability and $m_n$ given in~\eqref{eq:replicadetectionlast}. Different curves correspond to different crossover probabilities $\epsilon$.}
\label{fig:NRD} 
\end{figure}

Figure~\ref{fig:NRD} shows that as the crossover probability $\epsilon$ increases, \emph{i.e.,} the labeled correlated database $\mathbf{Y}$ gets noisier, the error probability decay becomes slower, leading to a worse detection performance, as expected. This is due to the following: Our simulation configuration leads to the following Binomial parameters $p_0$ (Eq.~\eqref{eq:p0}) and $p_1$ (Eq.~\eqref{eq:p1}):
\begin{align}
    p_0 &= 1-\frac{1}{|\mathcal{X}|}\\
    p_1 &= 1-(1-\epsilon)^2-\frac{\epsilon^2}{|\mathcal{X}|-1}
\end{align}
Consequently,
\begin{align}
    p_0-p_1 &= (1-\epsilon)^2+\frac{\epsilon^2}{|\mathcal{X}|-1}-\frac{1}{|\mathcal{X}|},
\end{align}
and $D(\tau\|p_0)$ and $D(1-\tau\|1-p_1)$ are decreasing functions of $\epsilon$ for any $\epsilon<1-\frac{1}{|\mathcal{X}|}$. Thus, we have a worse replica detection performance for a higher crossover probability $\epsilon$.

\subsection{Seeded Deletion Detection}
\label{subsec:SDDexp}
While our replica detection algorithm (Algorithm~\ref{alg:noisyreplicadetection}) works for a small column size as shown in Figure~\ref{fig:NRD}, our seeded deletion detection algorithm (Algorithm~\ref{alg:deletiondetection}) requires large column sizes. This is because the threshold $\hat{\tau}_n$ of Algorithm~\ref{alg:deletiondetection} is chosen based on its asymptotic properties and does not take estimated parameters into account, unlike the threshold $m_n\tau$ of Algorithm~\ref{alg:noisyreplicadetection}. Hence, while the outliers still exist in $\mathbf{L}(\Phi)$ (Eq.~\eqref{eq:crosshammingdistremapped}) for any useful remapping $\Phi$, the threshold becomes too large to distinguish the outlier from the rest of the entries.

\begin{figure}[t]
\centerline{\includegraphics[width=0.65\textwidth,trim={0cm 6cm 1cm 6cm},clip]{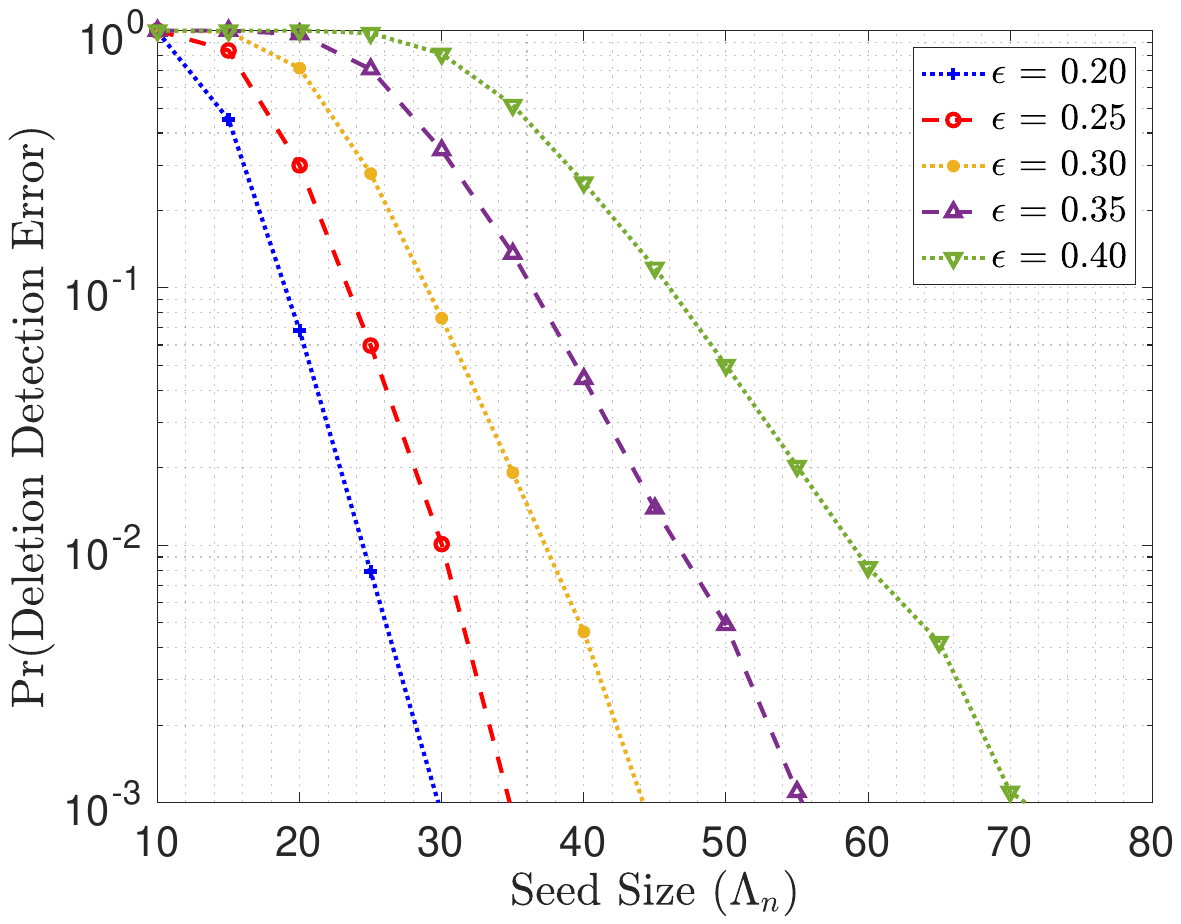}}
\caption{Probability of error of the modified seeded deletion detection algorithm (Algorithm~\ref{alg:modifieddeletiondetection}) vs. the seed size $\Lambda_n$ with $\Tilde{\tau}=1.5$ and $10^4$ trials. The y-axis is given in the logarithmic domain to validate the exponential relation between the deletion detection error probability and $\Lambda_n$ given by Lemma~\ref{lem: deletion detection}. Different curves correspond to different crossover probabilities $\epsilon$.}
\label{fig:SDD} 
\end{figure}

To overcome this problem, we propose a modification to Algorithm~\ref{alg:deletiondetection} based on the following observation: As shown for a family of distributions in~\cite{balakrishnan2008asymptotic}, the largest two order statistics of a sequence of $n$ \emph{i.i.d.} random variables typically have close values for large $n$. More formally, let $T_{(1),j}\le \dots\le T_{(n),j}$ be the order statistics of the $j$\textsuperscript{th} column of $\mathbf{M}(\Phi)$ (Eq.~\eqref{eq:absdeviations}) and define
\begin{align}
    R_{(i),j} &\triangleq \frac{T_{(i),j}}{T_{(i-1),j}},\hspace{1em} 2\le i\le n. \label{eq:ratioorderstats}
\end{align}
Then, for large $n$, we expect to have $T_{(n),j} = M_{r_j,j}(\Phi)$, where $r_j$ is the index of the $j$\textsuperscript{th} retained column (Eq.~\eqref{eq:rjdefn}), and
\begin{align}
    \lim\limits_{n\to\infty} R_{(n),j} &= \infty\\
    \lim\limits_{n\to\infty} R_{(n-1),j} &= 1
\end{align}
Based on this, we propose the following modification: First, we compute the sample means $\bar{R}_{(n)}$ and $\bar{R}_{(n-1)}$ of $\smash{\{R_{(n),j}\}_{j=1}^{\hat{K}_n}}$ and $\smash{\{R_{(n-1),j}\}_{j=1}^{\hat{K}_n}}$, respectively. Then, the algorithm declares $\Phi$ to be useful if 
\begin{align}
    \frac{\bar{R}_{(n)}}{\bar{R}_{(n-1)}}&\ge\Tilde{\tau}
\end{align}
for a threshold $\Tilde{\tau}$ above 1. We note that while any finite value of $\Tilde{\tau}$ that is larger than and bounded away from $1$ works in the asymptotic regime where $n\to\infty$, in the small column size regime, for each value of $n$, $\Tilde{\tau}$ is a heuristically chosen threshold slightly larger than 1, \emph{e.g.} $\Tilde{\tau} = 1.5$ for $n=100$.

After deciding that $\Phi$ is useful, the $j$\textsuperscript{th} retained column is inferred as:
\begin{align}
    \hat{I}_R(j) &= \argmax\limits_{i\in[n]} M_{i,j}(\Phi)
\end{align}
The resulting modified deletion detection algorithm is given in Algorithm~\ref{alg:modifieddeletiondetection}. From Figure~\ref{fig:SDD}, it can be seen that Algorithm~\ref{alg:modifieddeletiondetection} performs well for a small column size.

\begin{algorithm}[t]
\caption{Modified Distribution-Agnostic Seeded Deletion Detection Algorithm}\label{alg:modifieddeletiondetection}
\Input{$(\mathbf{G}^{(1)},\mathbf{G}^{(2)},\Tilde{\tau})$}
\Output{$\hat{I}_R$}
$\mathcal{S}(\mathcal{X})\gets$ SymmetryGroup($\mathcal{X}$)\;

\For{$s \gets 1$ \KwTo $|\mathcal{X}|!$}{
$\hat{I}_R \gets\varnothing$\;
$\Phi\gets\mathcal{S}(\mathcal{X})[s]$\Comment*[r]{Pick a remapping.} 
$\mathbf{L}(\Phi) \gets $ HammDist($\mathbf{G}^{(1)},\mathbf{G}^{(2)}(\Phi)$)\Comment*[r]{Eq. \eqref{eq:crosshammingdist}}
$\mu(\Phi)\gets$ SampleMean($\mathbf{L}(\Phi) $)\Comment*[r]{Eq.~\eqref{eq:samplemean}}
$\mathbf{M}(\Phi) \gets |\mathbf{L}(\Phi)-\mu(\Phi)|$\Comment*[r]{Eq. \eqref{eq:absdeviations}}
$\mathbf{R} \gets \textup{OrderStatisticRatios(}\mathbf{M}(\Phi) \textup{)}$\Comment*[r]{Eq. \eqref{eq:ratioorderstats}}
\eIf{\textup{mean(}$\mathbf{R}$\textup{[n][:])}$<\Tilde{\tau}$ \textup{mean(}$\mathbf{R}$\textup{[n-1][:])}}{
\eIf{$s\neq|\mathcal{X}|!$}{Skip to next $\Phi$\Comment*[r]{$\Phi$ is useless.}}{\Return ERROR \Comment*[r]{No useful $\Phi$ found.}}
}{
\For{$j\gets1$ \KwTo $\Tilde{K}_n$}{
$\hat{I}_R \gets \hat{I}_R \cup \{\argmax\limits_{i\in[n]} \mathbf{M}(\Phi)[i][j]\} $\;
}
\Return $\hat{I}_R$\;
}
}

\end{algorithm}

Comparing the different curves in Figure~\ref{fig:SDD}, we can conclude that as the crossover probability $\epsilon$ increases, the decay of the error probability becomes slower, indicating a worse detection performance. In turn, to achieve the same deletion detection error rate, we need a higher number of seeds in scenarios with a higher level of obfuscation.

\subsection{De-Anonymization Scheme}
\label{subsec:typicalityexp}
In this subsection, we evaluate the performance of the de-anonymization scheme (Algorithm~\ref{alg:matchingscheme}) proposed in Section~\ref{subsec:achievability}. While doing so, we decrease the column size from $n=100$ to $n=25$ due to the following considerations:
\begin{enumerate}[label=\textbf{(\roman*)}]
    \item As discussed in Section~\ref{subsec:achievability}, Algorithm~\ref{alg:matchingscheme} has a runtime of $\mathcal{O}(m_n^2 n)$. 
    \item As in the proof of Theorem~\ref{thm:mainresult} (Eq.~\eqref{eq:error}), the matching error probability decreases exponentially in the column size $n$. Hence to observe a non-trivial (non-zero) matching error, we must consider a row size $m_n$ exponential in $n$ (Definition~\ref{defn:dbgrowthrate}). This suggests, simulating Algorithm~\ref{alg:matchingscheme} in the non-trivial regime becomes computationally prohibitive, even for $n=100$.
\end{enumerate}

In addition to the modified deletion detection algorithm (Algorithm~\ref{alg:modifieddeletiondetection}), we make a slight modification to Algorithm~\ref{alg:matchingscheme}, described below to accommodate the small column size. After performing replica and deletion detection, and estimating $\hat{p}_{X,Y^S|S}$, instead of fixing an $\epsilon>0$ (that depends on $n$) and checking the $\epsilon$-joint-typicality of the row pairs from $\mathbf{X}$ and $\tilde{\mathbf{Y}}$, in Algorithm~\ref{alg:modifiedmatchingscheme}, we do the following:
\begin{enumerate}[label = \textbf{(\roman*)}]
    \item For each $(i,j)\in[m_n]^2$, compute
    \begin{align}
        \Delta_{i,j}&\triangleq \left|\hat{H}(X,Y^S|S)-\hat{H}_{i,j}(X,Y^S|S)\right|\label{eq:Deltaij}
    \end{align}
    where
    \begin{align}
        \hat{H}_{i,j}(X,Y^S|S)&\triangleq -\frac{1}{n} \log\hat{p}_{X^n,Y^{K_n}|S^n}(X^{n}_i,Y^{K_n}_j|\hat{S}^n)
    \end{align}
    \item For each $j\in[m_n]$, match the $j$\textsuperscript{th} row $Y^{K_n}_j$ of $\tilde{\mathbf{Y}}$ with the $l$\textsuperscript{th} row $X^n_l$ of {$\mathbf{X}$}, if
    \begin{align}
        l &= \argmin\limits_{i\in[m_n]} \Delta_{i,j}\label{eq:modifieddeltas}
    \end{align}
    assigning $\hat\sigma_n(l)=j$.
    \item If there exist distinct $j_1$, $j_2$ such that $l$ satisfies~\eqref{eq:modifieddeltas}, declare error and assign $\hat{\sigma}_n(l)=0$.
\end{enumerate}

We note that testing $\Delta_{i,j}$ against a threshold $\epsilon$ would correspond to the joint-typicality check as done in Algorithm~\ref{alg:matchingscheme}.

The experimental evaluation of the performance of Algorithm~\ref{alg:modifiedmatchingscheme} is given in Figure~\ref{fig:typicalitymatching}. As expected, we observe that a higher row size $m_n$ leads to an increased matching error. 

\begin{algorithm}[t]
\caption{Modified Distribution-Agnostic De-Anonymization Scheme}\label{alg:modifiedmatchingscheme}
\Input{$(\mathbf{X},\mathbf{Y},\mathbf{G}^{(1)},\mathbf{G}^{(2)})$}
\Output{$\hat{\sigma}_n$}
\textup{isReplica}$\gets$ Alg.\ref{alg:noisyreplicadetection}($\mathbf{Y}$)\;
$\mathbf{G}^{(2)}\gets$ removeExtraReplicas($\mathbf{G}^{(2)}$,isReplica)\;
\textup{isDeleted}$\gets$ Alg.\ref{alg:modifieddeletiondetection}($\mathbf{G}^{(1)},\mathbf{G}^{(2)}$,$\Tilde{\tau}$)\;
$\hat{S}^n \gets $ EstimateRepetitionPattern(isReplica,isDeleted)\;
$\hat{p}_{X,Y^S|S}\gets $ EstimateDist($\mathbf{G}^{(1)},\mathbf{G}^{(2)},\hat{S}^n$)\;
$\tilde{\mathbf{Y}}\gets$ MarkerAddition($\mathbf{Y},\hat{S}^n$)\;
\blue{/* Computing $\Delta_{i,j}$ Eq. \eqref{eq:Deltaij}*/}\\
$\mathbf{\Delta}\gets$ computeDeltas($\mathbf{X},\mathbf{Y},\hat{S}^n,\hat{p}_{X,Y^S|S}$)\;

\For{$i = 1$ \KwTo  \textup{rowSize($\mathbf{X}$)}}{
\textup{count}$\gets 0$\;
\For{$j = 1$ \KwTo  \textup{rowSize($\tilde{\mathbf{Y}}$)}}{
  \If{$i=\argmin\limits_{l\in[m_n]} \Delta_{l,j}$}{
    $\hat{\sigma}_n[i] \gets$ j\;
    count$\gets$ count + 1\;
  }
  
}
\If{\textup{count} $\neq 1$}{
$\hat{\sigma}_n[i] \gets 0$\Comment*[r]{Matching error.}
}
}

\end{algorithm}

\begin{figure}[t]
\centerline{\includegraphics[width=0.65\textwidth,trim={0cm 6cm 1cm 6cm},clip]{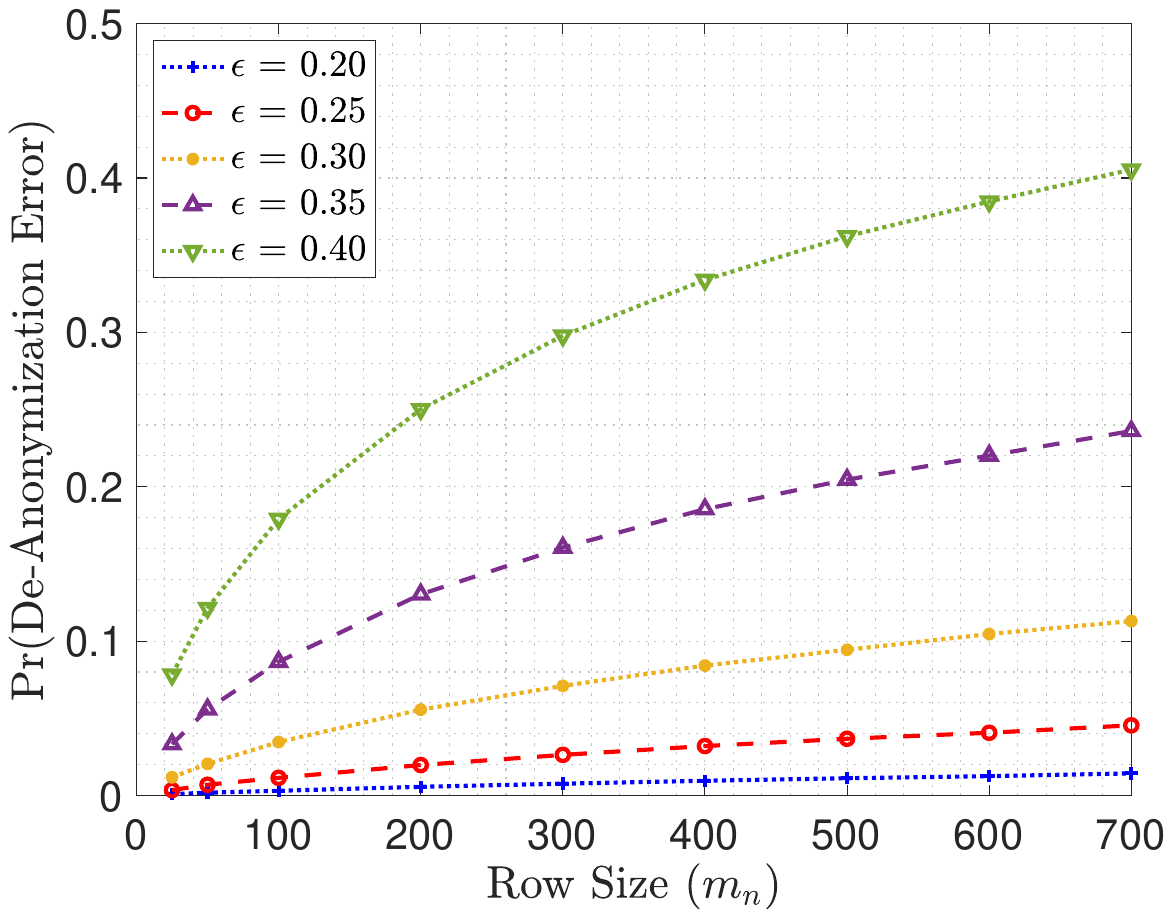}}
\caption{Probability of de-anonymization error of the modified de-anonymization scheme (Algorithm~\ref{alg:modifiedmatchingscheme}) vs. the row size $m_n$ with $n=25$, $\Lambda_n=25$ and $10^4$ trials. Different curves correspond to different crossover probabilities $\epsilon$.}
\label{fig:typicalitymatching} 
\end{figure}

Similar to Figures~\ref{fig:NRD} and \ref{fig:SDD}, Figure~\ref{fig:typicalitymatching} demonstrates the impact of the obfuscation on the de-anonymization performance. As the databases become more obfuscated, the de-anonymization performance degrades potentially an order of magnitude. Hence, we can conclude that the amount of obfuscation plays a crucial role in preserving privacy in databases.

\subsection{Histogram-Based Repetition Detection}
\label{subsec:HBDexp}

In this subsection, we evaluate the performance of the histogram-based repetition detection algorithm (Algorithm~\ref{alg:noiselessmatching}) of Section~\ref{subsec:noiseless} in the no-obfuscation setting. 

For uniform database distribution $X\sim$ Unif($[\mathcal{X}]$), we can obtain the following closed-form asymptotic expression for the repetition detection error probability of Algorithm~\ref{alg:noiselessmatching}:
\begin{prop}{\textbf{\cite[Propositon~4]{bakirtas2023database}}}
\label{prop:histogramuniform}
    Let $\xi_n$ denote the probability of the column histograms of $\mathbf{X}$ not being unique. If $X\sim\text{Unif}(\mathcal{X})$, then 
\begin{align}
    \xi_n=n^2 m_n^{\frac{1-|\mathcal{X}|}{2}} \left(4\pi \right)^{\frac{1-|\mathcal{X}|}{2}} |\mathcal{X}|^{\frac{|\mathcal{X}|}{2}} (1+o_{m_n}(1))(1-o_n(1))
\end{align}
\end{prop}

Proposition~\ref{prop:histogramuniform} states that for fixed $n$ and $|\mathcal{X}|$, $\xi_n\propto m_n^{\frac{1-|\mathcal{X}|}{2}}$, \emph{i.e.,} $\log \xi_n$ is linear with $\log m_n$ for large $n$ and $m_n$. More formally, 
\begin{align}
    \log \xi_n &= \frac{1-|\mathcal{X}|}{2} \log m_n + C_{|\mathcal{X}|} + 2 \log n + \zeta_{m_n,n}\label{eq:HBDlinear}
\end{align}
where 
\begin{align}
    C_{|\mathcal{X}|} &\triangleq \frac{1-|\mathcal{X}|}{2}\log 4\pi + \frac{|\mathcal{X}|}{2}\log |\mathcal{X}|
\end{align}
and $\zeta_{m_n,n}\to 0$ as $m_n,n\to\infty$. We note that \eqref{eq:HBDlinear} implies that the order relation given in Lemma~\ref{lem:histogram} is tight.

\begin{figure}[t]
\centerline{\includegraphics[width=0.65\textwidth,trim={0cm 6cm 1cm 6cm},clip]{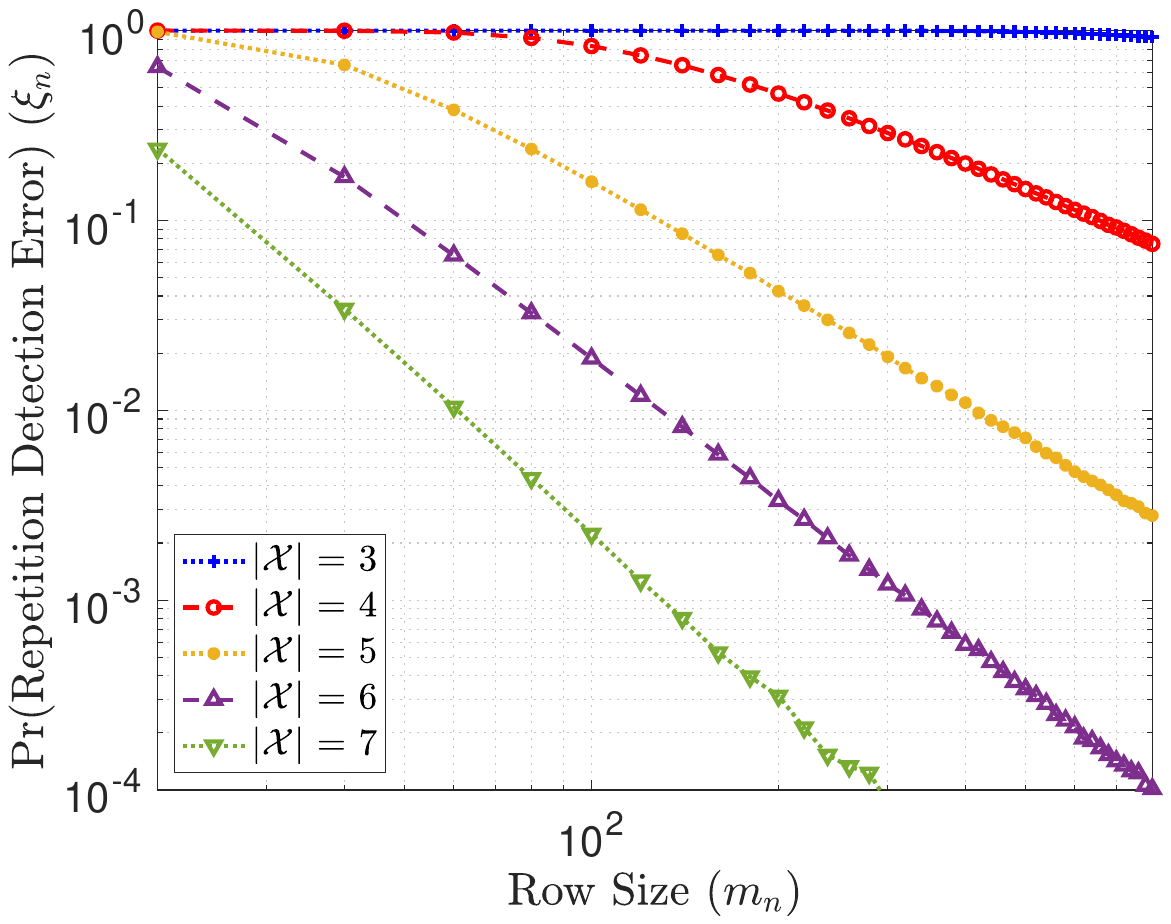}}
\caption{Probability of error of the histogram-based repetition detection~(Section~\ref{subsec:noiseless}) vs. the row size $m_n$ with $10^6$ trials. Both axes are given in the logarithmic domain to validate the linear relation given in~\eqref{eq:HBDlinear}.}
\label{fig:HBD} 
\end{figure}

From Figure~\ref{fig:HBD},  one can see that after some $m_n$, $\log \xi_n$ decays linearly with $\log m_n$ (excluding $|\mathcal{X}|=3$) with respective approximate slopes of $-1.40,-1.97,-2.51,-2.97$ demonstrating that~\eqref{eq:HBDlinear}, which predicts respective asymptotic slopes of $-1.50,-2,-2.50,-3$, holds even when $n=100$. Note the value of $m_n$ where $\xi_n$ starts decaying becomes smaller with increasing alphabet size $|\mathcal{X}|$. 

In Figure~\ref{fig:HBD}, as the alphabet $\mathcal{X}$ gets larger, the decay of the error probability is steeper. This is inherent to the nature of the histograms as in Proposition~\ref{prop:histogramuniform} and not an artifact of the simulation configuration chosen.

\section{Conclusion}
\label{sec:conclusion}
In this paper, we have investigated the distribution-agnostic database de-anonymization problem under synchronization errors, in the form of column repetitions, and obfuscation, in the form of noise on the retained entries of the database. We showed that through either modifications of the existing algorithms for the distribution-agnostic setting or novel ones, we can correctly detect the repetition pattern and perform de-anonymization, in both noisy/obfuscated and no-obfuscation settings. Interestingly, our results show that in terms of matching capacity, there is no penalty for not knowing the underlying distributions asymptotically. Our experimental results illustrate that our proposed algorithms or their slightly modified versions work in the non-asymptotic regime with small database sizes as well. Overall, our work provides insights into the practicality of distribution-agnostic database de-anonymization algorithms with theoretical guarantees.

\bibliography{references}

\begin{thebibliography}{10}
\providecommand{\url}[1]{#1}
\csname url@samestyle\endcsname
\providecommand{\newblock}{\relax}
\providecommand{\bibinfo}[2]{#2}
\providecommand{\BIBentrySTDinterwordspacing}{\spaceskip=0pt\relax}
\providecommand{\BIBentryALTinterwordstretchfactor}{4}
\providecommand{\BIBentryALTinterwordspacing}{\spaceskip=\fontdimen2\font plus
\BIBentryALTinterwordstretchfactor\fontdimen3\font minus \fontdimen4\font\relax}
\providecommand{\BIBforeignlanguage}[2]{{%
\expandafter\ifx\csname l@#1\endcsname\relax
\typeout{** WARNING: IEEEtran.bst: No hyphenation pattern has been}%
\typeout{** loaded for the language `#1'. Using the pattern for}%
\typeout{** the default language instead.}%
\else
\language=\csname l@#1\endcsname
\fi
#2}}
\providecommand{\BIBdecl}{\relax}
\BIBdecl

\bibitem{ohm2009broken}
P.~Ohm, ``Broken promises of privacy: Responding to the surprising failure of anonymization,'' \emph{UCLA L. Rev.}, vol.~57, p. 1701, 2009.

\bibitem{naini2015you}
F.~M. {Naini}, J.~{Unnikrishnan}, P.~{Thiran}, and M.~{Vetterli}, ``Where you are is who you are: User identification by matching statistics,'' \emph{IEEE Trans. Inf. Forensics Security}, vol.~11, no.~2, pp. 358--372, 2016.

\bibitem{datta2012provable}
A.~Datta, D.~Sharma, and A.~Sinha, ``Provable de-anonymization of large datasets with sparse dimensions,'' in \emph{International Conference on Principles of Security and Trust}.\hskip 1em plus 0.5em minus 0.4em\relax Springer, 2012, pp. 229--248.

\bibitem{narayanan2008robust}
A.~{Narayanan} and V.~{Shmatikov}, ``Robust de-anonymization of large sparse datasets,'' in \emph{Proc. of IEEE Symposium on Security and Privacy}, 2008, pp. 111--125.

\bibitem{sweeney1997weaving}
L.~Sweeney, ``Weaving technology and policy together to maintain confidentiality,'' \emph{The Journal of Law, Medicine \& Ethics}, vol.~25, no. 2-3, pp. 98--110, 1997.

\bibitem{takbiri2018matching}
N.~Takbiri, A.~Houmansadr, D.~L. Goeckel, and H.~Pishro-Nik, ``Matching anonymized and obfuscated time series to users’ profiles,'' \emph{IEEE Transactions on Information Theory}, vol.~65, no.~2, pp. 724--741, 2019.

\bibitem{cullina}
D.~{Cullina}, P.~{Mittal}, and N.~{Kiyavash}, ``Fundamental limits of database alignment,'' in \emph{Proc. of IEEE International Symposium on Information Theory (ISIT)}, 2018, pp. 651--655.

\bibitem{shirani8849392}
F.~{Shirani}, S.~{Garg}, and E.~{Erkip}, ``A concentration of measure approach to database de-anonymization,'' in \emph{Proc. of IEEE International Symposium on Information Theory (ISIT)}, 2019, pp. 2748--2752.

\bibitem{dai2019database}
O.~E. Dai, D.~Cullina, and N.~Kiyavash, ``{Database alignment with Gaussian features},'' in \emph{The 22nd International Conference on Artificial Intelligence and Statistics}.\hskip 1em plus 0.5em minus 0.4em\relax PMLR, 2019, pp. 3225--3233.

\bibitem{kunisky2022strong}
D.~Kunisky and J.~Niles-Weed, ``Strong recovery of geometric planted matchings,'' in \emph{Proc. of the 2022 Annual ACM-SIAM Symposium on Discrete Algorithms (SODA)}.\hskip 1em plus 0.5em minus 0.4em\relax SIAM, 2022, pp. 834--876.

\bibitem{tamir2023correlation}
R.~Tamir, ``{On correlation detection of Gaussian databases via local decision making},'' in \emph{2023 IEEE International Symposium on Information Theory (ISIT)}.\hskip 1em plus 0.5em minus 0.4em\relax IEEE, 2023, pp. 1231--1236.

\bibitem{bakirtas2021database}
S.~Bakirtas and E.~Erkip, ``Database matching under column deletions,'' in \emph{Proc. of IEEE International Symposium on Information Theory (ISIT)}, 2021, pp. 2720--2725.

\bibitem{bakirtas2022matching}
------, ``{Matching of Markov databases under random column repetitions},'' in \emph{2022 56th Asilomar Conference on Signals, Systems, and Computers}, 2022.

\bibitem{bakirtas2022seeded}
------, ``Seeded database matching under noisy column repetitions,'' in \emph{2022 IEEE Information Theory Workshop (ITW)}.\hskip 1em plus 0.5em minus 0.4em\relax IEEE, 2022, pp. 386--391.

\bibitem{bakirtas2023database}
------, ``Database matching under noisy synchronization errors,'' \emph{arXiv preprint arXiv:2301.06796}, 2023.

\bibitem{bakirtas2022database}
------, ``Database matching under adversarial column deletions,'' in \emph{2023 IEEE Information Theory Workshop (ITW)}.\hskip 1em plus 0.5em minus 0.4em\relax IEEE, 2023, pp. 181--185.

\bibitem{cormen2022introduction}
T.~H. Cormen, C.~E. Leiserson, R.~L. Rivest, and C.~Stein, \emph{{Introduction to Algorithms}}.\hskip 1em plus 0.5em minus 0.4em\relax MIT press, 2022.

\bibitem{cover2006elements}
T.~M. Cover, \emph{Elements of Information Theory}.\hskip 1em plus 0.5em minus 0.4em\relax John Wiley \& Sons, 2006.

\bibitem{shirani2021concentration}
F.~Shirani, S.~Garg, and E.~Erkip, ``A concentration of measure approach to correlated graph matching,'' \emph{IEEE Journal on Selected Areas in Information Theory}, vol.~2, no.~1, pp. 338--351, 2021.

\bibitem{morvai2005order}
G.~Morvai and B.~Weiss, ``{Order estimation of Markov chains},'' \emph{IEEE Transactions on Information Theory}, vol.~51, no.~4, pp. 1496--1497, 2005.

\bibitem{blischke1962moment}
W.~Blischke, ``{Moment estimators for the parameters of a mixture of two Binomial distributions},'' \emph{The Annals of Mathematical Statistics}, pp. 444--454, 1962.

\bibitem{moore2007basic}
D.~S. Moore and S.~Kirkland, \emph{The Basic Practice of Statistics}.\hskip 1em plus 0.5em minus 0.4em\relax WH Freeman New York, 2007, vol.~2.

\bibitem{anderson1957statistical}
T.~W. Anderson and L.~A. Goodman, ``{Statistical inference about Markov chains},'' \emph{The Annals of Mathematical Statistics}, pp. 89--110, 1957.

\bibitem{van2000asymptotic}
A.~W. Van~der Vaart, \emph{Asymptotic Statistics}.\hskip 1em plus 0.5em minus 0.4em\relax Cambridge University Press, 2000, vol.~3.

\bibitem{balakrishnan2008asymptotic}
N.~Balakrishnan and A.~Stepanov, ``Asymptotic properties of the ratio of order statistics,'' \emph{Statistics \& probability letters}, vol.~78, no.~3, pp. 301--310, 2008.

\bibitem{ash2012information}
R.~B. Ash, \emph{{Information Theory}}.\hskip 1em plus 0.5em minus 0.4em\relax Courier Corporation, 2012.

\bibitem{wasserman2004all}
L.~Wasserman, \emph{All of Statistics: A Concise Course in Statistical Inference}.\hskip 1em plus 0.5em minus 0.4em\relax Springer, 2004, vol.~26.

\end{thebibliography}
\bibliographystyle{IEEEtran}

\appendix
\subsection{Proof of Lemma~\ref{lem: replica detection}}
\label{proof: replica detection}
The estimator proposed in~\cite{blischke1962moment} works as follows: Define the $k$\textsuperscript{th} sample factorial moment $F_k$ as
\begin{align}
    F_k&\triangleq \frac{1}{K_n-1} \sum\limits_{j=1}^{K_n-1} \prod\limits_{i=0}^{k-1} \frac{W_j-i}{m_n-i},\hspace{1em} \forall k\in[m_n]\label{eq:factorialmoments}
\end{align}
and let 
\begin{align}
    U &\triangleq \frac{F_3-F_1 F_2}{F_2-F_1^2}
\end{align}
Then the respective estimators $\hat{p}_0$ and $\hat{p}_1$ for $p_0$ and $p_1$ can be constructed as:
\begin{align}
    \hat{p}_0 &= \frac{U+\sqrt{U^2-4 U F_1 + 4 F_2}}{2}\label{eq:paramest1}\\
    \hat{p}_1 &= \frac{U-\sqrt{U^2-4 U F_1 + 4 F_2}}{2}\label{eq:paramest2}
\end{align}
From~\cite{blischke1962moment}, we get $\hat{p}_i\overset{p}{\to}p_i$, $i=0,1$, and in turn $\tau\overset{p}{\to}\frac{U}{2}$. Thus for large $n$, $\tau$ is bounded away from $p_0$ and $p_1$. We complete the proof following the same steps taken in the proof of~\cite[Lemma~1]{bakirtas2023database}, which we provide below for the sake of completeness.

Let $A_j$ denote the event that $\smash{C^{(2)}_{j}}$ and $\smash{C^{(2)}_{j+1}}$ are noisy replicas and $B_j$ denote the event that the algorithm infers $\smash{C^{(2)}_{j}}$ and $\smash{C^{(2)}_{j+1}}$ as replicas. Via the union bound, we can upper bound the total probability of replica detection error $\kappa^{(1)}_n$ as
\begin{align}
    \kappa^{(1)}_n&= \Pr(\bigcup\limits_{j=1}^{{K_n}-1} E_j)\label{eq:kappa1defn}\\
    &\le \sum\limits_{j=1}^{{K_n}-1} \Pr(A_j ^c) \Pr(B_j|A_j ^c)+ \Pr(A_j)  \Pr(B_j^c|A_j)\label{eq:replicadetectionbound}
\end{align}

Observe that conditioned on $A_j^c$, $W_j\sim\text{Binom}\smash{(m_n,p_0)}$ and conditioned on $A_j$, $W_j\sim\text{Binom}(m_n,p_1)$. Then, from the Chernoff bound~\cite[Lemma 4.7.2]{ash2012information}, we get
\begin{align}
    \Pr(B_j|A_j ^c)&\le 2^{-m_n D\left(\tau\|\smash{p_0}\right)}\label{eq:chernoff1}\\
    \Pr(B_j^c|A_j)&\le 2^{-m_n D\left(1-\tau \|1-p_1\right)}\label{eq:chernoff2}
\end{align}

Thus, through the union bound, we obtain
\begin{align}
    \kappa^{(1)}_n&\le ({K_n}-1)\Big[ 2^{-m_n D\left(\tau\|\smash{p_0}\right)}+ 2^{-m_n D\left(1-\tau\|1-p_1\right)}\Big]\label{eq:replicadetectionlast}
\end{align}

Since the RHS of \eqref{eq:replicadetectionlast} has $2{K_n}-2=\mathcal{O}(n)$ terms decaying exponentially in~$m_n$, for any $m_n=\omega(\log n)$ we have 
\begin{align}
    \kappa^{(1)}_n \to 0 \:\text{ as } n\to\infty.
\end{align}

Observing that $n\sim \log m_n$ concludes the proof.
\qed
\subsection{Proof of Lemma~\ref{lem: deletion detection}}
\label{proof: deletion detection}
    For now, suppose that $\Phi$ is a useful remapping.
    Using Chebyshev's inequality~\cite[Theorem 4.2]{wasserman2004all} it is straightforward to prove that for any $\epsilon_n>0$
\begin{align}
    \gamma&\triangleq \Pr(|\mu(\Phi)-\Lambda_n q_0(\Phi)|>\Lambda_n {\epsilon}_n) \\ &= \mathcal{O}\left(\frac{1}{\Tilde{K}_n n \Lambda_n {\epsilon}_n}\right)\label{eq:alpha}
\end{align}
First, let
\begin{align}
    I_R&=\{r_1,\dots,r_{\tilde{K}_n}\} \label{eq:rjdefn}
\end{align}
and note ${L_{r_j,j}(\Phi)\sim\text{Binom}(\Lambda_n,q_1(\Phi))}$. Thus, from the Chernoff bound~\cite[Lemma 4.7.2]{ash2012information} we get
\begin{align}
    \beta_{r_j,j}&\triangleq\Pr(|L_{r_j,j}(\Phi)-\Lambda_n q_1(\Phi)|\ge {\epsilon}_n \Lambda_n)\\
    &\le 2^{-\Lambda_n D(q_1(\Phi)-{\epsilon}_n\|q_1(\Phi))}
    + 2^{-\Lambda_n D(1-q_1(\Phi)-{\epsilon}_n\|1-q_1(\Phi))}.\label{eq:beta}
\end{align}

Now, for notational brevity, let 
\begin{align}
    f(\epsilon) &\triangleq D(q-\epsilon\|q)
\end{align}
Then, one can simply verify the following
\begin{align}
    f^\prime(\epsilon) &= \log\frac{q}{1-q}-\log\frac{q-\epsilon}{1-q-\epsilon}\\
    f^{\prime\prime}(\epsilon)&=\frac{1}{\log e} \left[\frac{1}{q-\epsilon} + \frac{1}{1-q+\epsilon}\right]
\end{align}
Similarly, letting
\begin{align}
    g(\epsilon) &\triangleq D(1-q-\epsilon\|1-q)
\end{align}
we get
\begin{align}
    g^\prime(\epsilon) &= \log\frac{1-q}{q}-\log\frac{1-q-\epsilon}{q+\epsilon}\\
    g^{\prime\prime}(\epsilon)&=\frac{1}{\log e} \left[ \frac{1}{1-q-\epsilon} + \frac{1}{q+\epsilon}\right]
\end{align}
Observing that 
\begin{align}
    f(0)&=f^\prime(0)=0\\
    g(0)&=g^\prime(0)=0
\end{align}
and performing second-order MacLaurin series expansions on $f$ and $g$, we get for any $\epsilon<1$
\begin{align}
    f(\epsilon) &= c(q) \epsilon^2 + \mathcal{O}(\epsilon^3)\\
    g(\epsilon) &= c(q) \epsilon^2  + \mathcal{O}(\epsilon^3)
\end{align}
where
\begin{align}
    c(q)&\triangleq\frac{1}{\log e} \left[\frac{1}{q} + \frac{1}{1-q}\right]
\end{align}

Now, let $\Lambda_n = \Gamma_n \log n$ and $\epsilon_n =\Gamma_n^{-\nicefrac{1}{3}}$ and pick the threshold as $\hat{\tau}_n = 2\Lambda_n \epsilon_n$. Observe that since $\Gamma_n=\omega_n(1)$, we get
\begin{align}
    \hat{\tau}_n &= 2 \Lambda_n \epsilon_n  = o_n(\Lambda_n)\\
    \Lambda_n \epsilon_n^2 &= \Gamma_n^{\nicefrac{1}{3}} \log n = \omega_n(\log n)
\end{align}
Then, we have
\begin{align}
    \beta_{r_j,j} &\le 2^{1-\Lambda_n (c(q_1(\Phi)) \epsilon_{n}^2+\mathcal{O}(\epsilon_n^3))}\\
    &= 2^{1- c(q_1(\Phi))  \Gamma_n^{\nicefrac{1}{3}} \log n+\mathcal{O}(\epsilon_n^3))}
\end{align}
Note that with probability at least $1-\gamma-\beta_{r_j,j}$ we have
\begin{align}
    |\mu(\Phi)-\Lambda_n q_0(\Phi)|&\le \Lambda_n\epsilon_n\\
    |L_{r_j,j}(\Phi)-\Lambda_n q_1(\Phi)|&\ge \Lambda_n\epsilon_n
\end{align}
From the triangle inequality, we have
\begin{align}
    M_{r_j,j}(\Phi) &=|L_{r_j,j}(\Phi)-\mu(\Phi)|\\ &\ge \Lambda_n (|q_1(\Phi)-q_0(\Phi)|-2\epsilon_n)\\
    &\ge \hat{\tau}_n
\end{align}
for large $n$.
Therefore, from the union bound we have
\begin{align}
    \Pr(\exists j\in[\tilde{K}_n]:M_{r_j,j}(\Phi)\le \hat{\tau}_n)
    &\le \gamma + \sum\limits_{j=1}^{\tilde{K}_n} \beta_{r_j,j}\\
    &= \gamma + 2^{\log \tilde{K}_n - 1- c(q_1(\Phi)) \Gamma_n^{\nicefrac{1}{3}} \log n +\mathcal{O}(\epsilon_n^3))}
\end{align}
Since $\tilde{K}_n\le n$ and $\Lambda_n=\omega_n(\log n)$, we have
\begin{align}
   \lim\limits_{n\to\infty} \log \tilde{K}_n -c(q_1(\Phi)) \Gamma_n^{\nicefrac{1}{3}} \log n
   &= -\infty
\end{align}
Thus we have
\begin{align}
    \lim\limits_{n\to\infty} \Pr(\exists j\in[\tilde{K}_n]:M_{r_j,j}\le \hat{\tau}_n) &= 0
\end{align}

Next, we look at $i\neq r_j$. Repeating the same steps above, we get
\begin{align}
    \beta_{i,j}&\triangleq\Pr(|L_{i,j}(\Phi)-\Lambda_n q_0(\Phi)|\ge \epsilon_n \Lambda_n)\\
    &\le 2^{-\Lambda_n D(q_0(\Phi)-\epsilon_n\|q_0(\Phi))}
    + 2^{-\Lambda_n D(1-q_0(\Phi)-\epsilon_n\|1-q_0\Phi))}\\
    &= 2^{1- c(q_0(\Phi)) \Gamma_n^{\nicefrac{1}{3}} \log n+\mathcal{O}(\epsilon_n^3))}
\end{align}
Again, from the triangle inequality, we get
\begin{align}
    M_{i,j}(\Phi) &= 
    |L_{i,j}(\Phi)-\mu(\Phi)|\\
    &\le 2\epsilon_n\\&= \hat{\tau}_n
\end{align}
From the union bound, we obtain
\begin{align}
    \Pr(\exists j\in[\Tilde{K}_n]\:\exists i\in[n]\setminus\{r_j\}: M_{i,j}(\Phi)\ge \hat{\tau}_n) 
    &\le \gamma + \sum\limits_{j=1}^{\tilde{K}_n}\sum\limits_{i\neq r_j} \beta_{i,j}\\
    &\le \gamma + n^2 2^{1- c(q_0(\Phi))  \Gamma_n^{\nicefrac{1}{3}} \log n+\mathcal{O}(\epsilon_n^3))}
\end{align}
Since $\Lambda_n=\omega(\log n)$, as $n\to\infty$ we have 
\begin{align}
   \Pr(\exists j&\in[\Tilde{K}_n]\:\exists i\in[n]\setminus\{r_j\}:M_{i,j}(\Phi)\ge \hat{\tau}_n)\longrightarrow0.
\end{align}
Thus, for any useful remapping $\Phi$, the misdetection probability decays to zero as $n\to \infty$.

For any remapping $\Phi$ that is not useful, following the same steps, one can prove that
\begin{align}
    \Pr(\text{Remapping }\Phi \text{ is declared useful,} &\text{ even though it is not.})\notag\\
    &\le \gamma + \sum\limits_{i=1}^{n}\sum\limits_{j=1}^{\Tilde{K}_n} \Pr(M_{i,j}\ge \epsilon_n \Lambda_n)\\
    &\le \gamma + n^2 2^{1- c(q_0(\Phi))  \Gamma_n^{\nicefrac{1}{3}} \log n+\mathcal{O}(\epsilon_n^3))}\\
    &= o_n(1)
\end{align}
Since $|\mathcal{S}(\mathcal{X})|=|\mathcal{X}|!=\mathcal{O}_n(1)$, we have
\begin{align}
    \sum\limits_{\Phi:\text{not useful}} \Pr(\Phi \text{ is declared useful.})=o_n(1)
\end{align}
concluding the proof.
\qed
\subsection{Proof of Theorem~\ref{thm: noiseless}}
\label{proof: noiseless}
Let $\epsilon>0$ and denote by $\xi_n$ the error probability of the histogram-based repetition algorithm. Denote the $\epsilon$-typical set of sequences (with respect to $p_X$) of length $k=(1-\delta-\epsilon)n$ by $\smash{A_{\epsilon}^{(k)}(X)}$ and the pairwise collision probability between $X_1^n$ and $X_i^n$, given $\hat{K}_n=k$ and $\smash{X_1^n\in A_\epsilon^{(k)}(X)}$, by $P_{col,i}(k)$.
Since additional columns in $\bar{\mathbf{Y}}$ would decrease the pairwise collision probability between independent rows, we have
\begin{align}
    P_{col,i}(\hat{K}_n)&\le P_{col,i}(k),\hspace{1em} \forall\hat{K}_n\ge k
\end{align}
Since we perform exact row matching and the rows of $\bar{\mathbf{X}}$ are independent, we have
\begin{align}
    P_{col,i}(k) &= \Pr(\bar{X}_i^k=\bar{X}_1^k|\bar{X}_1^k\in A_{\epsilon}^{(k)}(X))\\
    &\le 2^{-k(H(X)-\epsilon)}
\end{align}
Thus, we can bound the probability of error $P_e$ as
\begin{align}
    P_e 
    &\le \sum\limits_{i=2}^{m_n} P_{col,i}(k)+\epsilon+\kappa_n+\xi_n\\
    &\le 2^{n R}  2^{-k(H(X)-\epsilon)}+\epsilon+\kappa_n+\xi_n
\end{align}
where $\kappa_n=\Pr(\hat{K}_n<k)$. Since $m_n$ is exponential in $n$, by Lemma~\ref{lem:histogram}, ${\xi_n\to0}$ as ${n\to\infty}$. Furthermore, $\hat{K}_n$ is a Binom($n,1-\delta$) random variable and from law of large numbers ${\kappa_n\to0}$ as ${n\to\infty}$. Thus $ P_e\le \epsilon$ as $n\to\infty$ if
\begin{align}
    R<(1-\delta-\epsilon)H(X).
\end{align}
Thus, we can argue that any database growth rate $R$ satisfying
\begin{align}
    R<(1-\delta)H(X)
\end{align}
is achievable, by taking $\epsilon$ small enough.
\qed

\end{document}